\documentclass[11pt]{article}
\usepackage{epsfig}
 \usepackage{amsmath}
 \usepackage{amsfonts}
 \usepackage{amssymb}
 \usepackage{amsthm}

 \usepackage{setspace}
 \usepackage{fullpage}

\usepackage{epsfig}
\usepackage{wrapfig}

\usepackage{enumerate}

\usepackage{soul,xcolor}

 \newcommand{\hide}[1]{}

\newcommand{\Xomit}[1]{}

 \def\qed{\ifmmode$\blacksquare$\else{\unskip\nobreak\hfil
 \penalty50\hskip1em\null\nobreak\hfil$\blacksquare$
 \parfillskip=0pt\finalhyphendemerits=0\endgraf}\fi\vspace{0.3cm}}

 \newtheorem{theorem}{Theorem}[section]
 \newtheorem{lemma}{Lemma}[section]
 \newtheorem{observation}[lemma]{Observation}
 
 \newtheorem{definition}[lemma]{Definition}
 \newtheorem{corollary}[lemma]{Corollary}
\newtheorem{fact}[lemma]{Fact}

 \newcommand{\vone}{\vspace{.1in}}
 \newcommand{\vhalf}{\vspace{.05in}}

 \def\eps{\epsilon}

 \newcommand{\ceil}[1]{\left\lceil{#1}\right\rceil}

\usepackage{comment}
\excludecomment{ignore}

\begin{document}

\title{\Large Resource Oblivious Sorting on Multicores \footnote{A preliminary version of this
paper appeared in {\it Proc. International Colloquium of Automata, Languages and Programming (ICALP)}, Track A, Springer LNCS Volume 6198, pp. 226-237, 2010.}}

\author{Richard Cole ~\thanks{Computer Science Dept., Courant Institute
of Mathematical Sciences, NYU, New York, NY 10012.
Email: {\tt cole@cs.nyu.edu}.
This work was supported in
 part by NSF Grants CCF-0830516, CCF-1217989, and CCF-1527568.}
\and Vijaya Ramachandran~\thanks
{Dept. of Computer Sciences, University of Texas, Austin, TX 78712. Email:
 {\tt vlr@cs.utexas.edu}. This work was supported in
 part by NSF Grants CCF-0830737 and CCF-1320675.}
 }

\maketitle

\begin{abstract}

We present a deterministic sorting algorithm,
SPMS (Sample, Partition, and Merge Sort), that interleaves
the partitioning of a sample sort with merging. Sequentially, it sorts $n$
elements in $O(n \log n)$ time cache-obliviously with an optimal
number of cache misses. The parallel complexity (or critical path
length) of the algorithm is $O(\log n \cdot \log\log n)$, which improves
on previous bounds for optimal cache oblivious sorting.
The algorithm also has low false sharing costs.
When scheduled by a work-stealing scheduler in a multicore
computing environment with a global shared memory and $p$ cores,
each having a cache of size $M$ organized in blocks of size $B$,
the costs of  the additional
cache misses and false sharing misses due to this parallel execution are
bounded by the cost of $O(S\cdot M/B)$ and  $O(S \cdot B)$ cache misses
respectively, where $S$ is the number of steals performed during the
execution. Finally, SPMS is resource oblivious in
that the dependence on machine parameters appear only in the analysis
of its performance, and
not within the algorithm itself.
\end{abstract}

\section{Introduction}\label{sec:intro}

We present a parallel sorting algorithm, which we call
{\it  Sample, Partition, and Merge Sort (SPMS)}. It has a critical
path length of $O(\log n\log\log n)$ and performs optimal $O(n\log
n)$ operations with optimal sequential cache misses.
Further, it has low caching and false sharing overhead when scheduled by a
work stealing scheduler
on a multicore with private caches, and is resource oblivious, i.e.,
it uses no machine parameters in its specification.

At the heart
of the sorting algorithm is a recursive multi-way merging
procedure. 
It creates its recursive subproblems using a sample sort methodology.
We view SPMS as interleaving a
merge sort with a sample sort in a natural way.
At this high level the structure of SPMS is similar to Sharesort~\cite{CP93}, 
a sorting algorithm developed for
the hypercube. 
In Section~\ref{sec:sort-start}, after we have given a high-level specification of SPMS,
we elaborate on the similarities of and differences between SPMS and Sharesort.
Sharesort was also adapted to several multi-level storage models 
in~\cite{AP94}. 

\vone
\noindent
{\bf Previous Work.}
Sorting is a
fundamental algorithmic
problem, and has been studied extensively. For our purposes, the
most relevant results are sequential cache-oblivious sorting, for which
provably optimal algorithms are
given in~\cite{FLPR99}, optimal sorting
algorithms addressing pure parallelism~\cite{AKS83,Co88}, and recent
work on multicore sorting \cite{BC+08,AGN08,BGS09,Va08}.
When considering the parallelism in an algorithm, we will interchangeably
use the terms critical pathlength, span, and parallel time to denote the number
of parallel steps in the algorithm when it is executed with a sufficiently large
number of processors to exploit all of the parallelism present in it.

The existing multicore algorithms take two main approaches.
The first is merge sort \cite{AGN08,BC+08}, either simple
or the pipelined method from \cite{Co88}.
The second is deterministic sampling \cite{Va08,BGS09}: this approach
splits the input into subsets, sorts the subsets,
samples the sorted subsets,
sorts the sample, partitions about a subsample,
and recursively sorts the resulting sets.
Our algorithm can be viewed as applying this approach to the
problem of merging a suitable number of sorted sets,
which eliminates the need for the first two steps, resulting in
a smaller critical pathlength.

More specifically, the algorithm in  \cite{BC+08} is a simple
multicore mergesort; it
has polylog parallel time, and good, though not
optimal cache efficiency; it is cache-oblivious for private
caches (the model we consider in this paper).
The algorithm in \cite{AGN08}
achieves the optimal caching bound
on an input of length $n$,
with $O(\log n)$ parallel
time (modulo dependence on cache parameters), but it is both cache-aware
and core-aware; this algorithm is based on \cite{Co88}.
The algorithm in \cite{Va08} is designed for a BSP-style version of a
cache aware,
multi-level multicore. It uses a different collection of parameters,
and so it is difficult to compare with it directly.
The algorithm in~\cite{BGS09} recursively sorts the samples,
achieving $O(\log^2 n)$ critical
path length deterministically, and $O(\log^{1.5} n)$ in a randomized version.

\vone
\noindent
{\bf Our Results.}
Our main results are stated in the following two theorems.
Theorem~\ref{thm:fs-full}, which applies to
any work-stealing scheduler, is established in
Section~\ref{sec:fs-estack}, and Theorem~\ref{thm:rws}, which applies to
the randomized work-stealing scheduler
(RWS), is established in Section~\ref{sec:perf}.
These theorems use the following notation: $M$ is the cache size, $B$ is the block size,
a {\it steal} is the process of transferrring a parallel task from the processor that generated it
to another processor that will start its execution, $p$ is the number of processors, $b$ is
the cost of a cache miss, and $s \geq b$ is the cost of a steal.

\begin{theorem}\label{thm:fs-full}
There is an implementation of the SPMS algorithm
that performs $O(n\log n)$ operations,
runs in parallel time $O(\log n \log \log n)$,
and, under work-stealing,
incurs $O(\frac nB \frac{\log n}{\log M} + \frac MB \cdot S)$ cache misses,
and incurs a total cost due to false sharing
 bounded by
the cost of
$O(S \cdot B) = O(\frac MB \cdot S)$ cache misses,
where $S$ is the number of steals.
These results assume a tall cache, i.e., $M = \Omega (B^2)$.
\end{theorem}

\begin{theorem}\label{thm:rws}
W.h.p. in $n$, the
execution time
of SPMS with the RWS scheduler,
when including both cache miss cost and false sharing cost is optimal for
\[
p = O\left(\frac {n}{M \log M} \cdot \min \left\{\frac{1}{\log\log n},~\frac{s}{b} \cdot \frac{\log B}{B}\right\}\right)
\]
if
$s = O(b\cdot \frac MB)$
and $M= \Omega(B^2)$.
\end{theorem}

\vhalf
\noindent
{\bf Roadmap}.
In Section~\ref{sec:sort-start} we describe the basic SPMS sorting algorithm,
and show that it has a fairly simple implementation with
optimal work (i.e., sequential running time) and
$O(\log n \cdot \log\log n)$ critical pathlength, and also a variant
 with optimal work and optimal sequential cache
complexity.  These
two sets of bounds are obtained with different implementations.
In Section~\ref{sec:co-parallel}, we present an implementation that
achieves all three bounds simultaneously. This section follows
Section~\ref{sec:model}, which describes the multicore caching
model, together with some basic methodology used in prior work
for developing efficient multicore
algorithms~\cite{BC+08,AGN08,CR08,BGS09,CR12}, and some basic
framework from~\cite{CR12} relevant to the multicore caching
model. Then, in Sections~\ref{sec:fs-misses} and \ref{sec:fs-estack}, we
address the overhead of
{\it false sharing}, a phenomenon that is inherent in a resource
oblivious multicore setting: Section~\ref{sec:fs-misses} bounds
the false sharing costs in SPMS when there is no deallocation of space
(as in, e.g., a functional programming setting), and Section~\ref{sec:fs-estack}
bounds the cost of additional fs misses that occur on the execution stack
when space is re-used there; the background for this setting is given
in Section~\ref{sec:exec-stack}. The results in
Sections~\ref{sec:fs-misses}-\ref{sec:fs-estack} build on our earlier work on false sharing in~\cite{CR12}
for hierarchical balanced parallel (HBP) computations; however, since
SPMS is not exactly an HBP computation due to variations in sub-problem sizes,
we present a self-contained analysis of the false sharing results for SPMS.
In Section~\ref{sec:perf}
we analyze the performance of
our final SPMS implementation
(given in Section~\ref{sec:fs-misses})
as a function of
the number of parallel tasks scheduled by the scheduler.
We also derive specific results for the Randomized Work Stealing (RWS)
scheduler, using some already known
performance results.
We conclude in Section \ref{sec:discussion} with a summary, and an
avenue for further research.

\section{SPMS, A New Deterministic Sample, Partition, and Merge Sort}
\label{sec:sort-start}

There are two main approaches used in parallel algorithms for sorting
(ignoring methods used in sorting networks): Sorting by merging, based on
merge-sort, and sorting by sampling and partitioning, called
Sample Sort or Distribution Sort.
Both methods can also achieve good cache-efficiency.

In parallel merge sort, the input is grouped into two or more sets, which
are sorted recursively, and the sorted sequences are then merged by a
parallel merging algorithm. If the merge is binary then we obtain a
parallelization of merge-sort, which can achieve optimal $O(\log n)$
parallel time \cite{Co88}, though it is not known how
to achieve optimal
cache-oblivious cache-efficiency. Multi-way merge can achieve optimal
cache-oblivious cache-efficiency, but
it is not known whether
very high levels of parallelism
can be achieved
when cache-oblivious cache-efficiency is also desired.

In sample sort, a small sample of the input
is taken and then sorted, typically by a simple parallel
algorithm, and not recursively since the sample size is small relative to the input size.
The sorted sample elements are used to partition the input elements into
a linearly ordered sequence of subsets of the input elements, and finally
each of these subsets is sorted recursively. Sample (or distribution)
sort can achieve optimal cache-oblivious cache-efficiency.

In sample sort, all of the recursively sorted subsets
have approximately the same
size, and this is achieved by having the ranks of the samples spaced fairly
evenly among the input elements. This is readily
ensured by choosing a random sample of the desired size.

In deterministic sample sort, the sample is usually chosen by partitioning
the input into a collection of sets, sorting each set (recursively), then
selecting every $k$-th element (for a suitable $k$) as a sample element, and
then sorting the resulting samples from the different sets by a simple
algorithm. Often, `oversampling' is employed, where every $l$-th element
in the sorted sample is used in the actual sample. Oversampling with
suitable parameters ensures that the partitioned subsets will have
approximately the same size, 
as in~\cite{AACS87}, for example.

As noted earlier,
after the sorted sample is obtained,
the input is partitioned into a linearly ordered sequence
of subsets, each of which is sorted recursively. However, in deterministic
sample sort, this is somewhat wasteful, since each set (from which the sample
elements were chosen) was sorted in order to obtain evenly spaced
sample elements, yet the sorted information in these sets is
not used once the samples are selected.

Our 
algorithm, Sample, Partition and Merge Sort (SPMS), uses a variant
of the sample and partition method of deterministic sample sort,
but it maintains the
sorted information from each original sorted subsequence within each
linearly ordered subset. Hence, the recursive problem to be solved is not
a recursive sort, but a recursive multi-way merge.
Accordingly, we generalize the problem being solved as follows.
Its input $A$ comprises up to
$r$ sorted lists  of total length
$n \le 3\cdot r^c$ where $c\ge 6$ is a constant,
and the task is to merge these lists.
We let SPMS$(A;n,r)$ denote this problem instance.
A sort is executed by setting $r=n$, i.e.\ the input comprises
$n$ lists each of length 1.

A previous sorting algorithm that 
uses deterministic sampling of sorted sequences and then
merges the sorted sequences 
is Sharesort~\cite{CP93}, which was designed for sorting
on hypercubic networks and runs in such networks in $O(\log n \cdot (\log \log n)^2)$ parallel time.
At a high
level SPMS is similar to Sharesort though, in contrast to Sharesort, the sorted lists in our method
can have widely differing sizes.
Below, after we give the high-level SPMS algorithm,
we discuss 
how these two algorithms differ.
Also, our methods are
tailored to a shared-memory multicore with private caches 
in a resource-oblivious setting 
(see Section~\ref{sec:model}).

The SPMS algorithm performs two successive collections of
recursive $\sqrt{r}$-way merges, each merge
being of size at most $3 \cdot r^{c/2}$.
To enable this, suitable samples of the input lists will
be sorted by a logarithmic time procedure, which then allows the
original problem to be partitioned into smaller subproblems that
are merged recursively.
Here is the high-level algorithm.

\vone
\noindent
{\bf SPMS}$(A; n,r)$

\noindent
{\bf Input.} $A$ is a collection of $r$ sorted lists $ L_1, L_2, \cdots, L_r $, of total length
$n \le 3 \cdot r^c$.

\noindent
{\bf Output.} The elements of $A$ rearranged in sorted order.

\vhalf
\noindent
{\bf Step 0.} {\bf if} $n \leq 24$ then sort with any sequential  $O(n\log n)$ time algorithm {\bf else} perform
Steps 1--3.

\vhalf
\noindent
{\bf Step 1.} {\bf If} $n > 3 r^{c/2}$ {\bf then}
{\it Deterministic Sample and Partition:}

\noindent
$~~~~$  {\bf 1a.} {\it Sample.} Form a  sample $S$ of every $r^{(c/2)-1}$-th element in each of the sorted lists in $A$, and sort $S$.\\
$~~~~$ {\bf 1b.} {\it Partition.}
Form a pivot set $P$ containing every $2r$-th element in the sorted $S$ (and for convenience add a zero'th pivot with value $-\infty$ and
a new largest pivot with value $\infty$). Use $P$ to  partition the original problem $A$
into $k= O(n/r^{\frac{c}{2}})$
disjoint merging subproblems, $A_1, A_2, \cdots,
A_k$, each comprising at most
$r$ sorted lists, with the items in $A_i$ preceding those in $A_{i+1}$, for $1 \leq i <k$.
Thus, the items in $A_i$ are those elements in $A$ with values at
least as large as the $(i-1)$st pivot, and less than the value of the
$i$-th pivot.
We show below that each $A_i$ contains at most $3 c\cdot r^{c/2}$ elements.

\vhalf
\noindent
{\bf Step 2.} {\it First Recursive Multi-way Merge:}\\
For each subproblem $A_i$, group its lists into
disjoint subsets of at most $\sqrt{r}$ lists,
and merge the lists in each group.
As $A_i$ contains at most
$3\cdot r^{\frac{c}{2}}$
items, this bound applies
to each of the individual groups too.
Thus the $\sqrt{r}$-way merge in each group can be performed recursively.
The output, for each subproblem $A_i$, is a
collection of at most
$\sqrt r$ sorted lists of total length at most
$3\cdot r^{\frac{c}{2}}$.
Thus Step 2 is the following.

\noindent
{\bf 2. for} $i=1$ {\bf to} $k$ {\bf do}\\
$~~~~$ {\bf 2a.} Form disjoint $A_{ij}$, with $j \leq \sqrt r$, where each  $A_{ij}$ contains at most $\sqrt r$ of the lists in $A_i$;\\
$~~~~~~~~~~~$ let $|A_{ij}| = n_{ij}$, for each $j$.\\
$~~~~$ {\bf 2b. for} each subproblem $A_{ij}$ {\bf do} SPMS$(A_{ij}; n_{ij}, \sqrt r)$

\vhalf
\noindent
{\bf Step 3.} {\it Second Recursive Multi-way Merge.}\\
 {\bf for}  each subproblem $A_i$ {\bf do} recursively merge
the $\sqrt r$ sorted lists computed in Step 2.\\
Let $|A_i| = n_i$ for each $i$.
Step 3 is the following.

\noindent
{\bf 3. for}  $i=1$ {\bf to} $k$ {\bf do} SPMS$(A_i; n_i, \sqrt r)$.

\vhalf
\noindent
{\bf Comment}. 1. The subproblems in Step 2 may have quite different sizes,
for these depend on the exact positions of the sampled items, which
the algorithm cannot completely control.
This, incidentally, is the reason for focussing on the more general
merging problem with possibly unequal sized sorted lists.

2. In contrast, initially 
Sharesort~\cite{CP93} creates equal-sized merging subproblems 
by using recursive sorts, and maintains roughly equal-sized merging subproblems
by grouping subproblems as needed.
We could also follow this approach, but then 
instead of needing one partitioning of
the items in Step 2a, 
we would need a second partitioning at the start of Step 3,
and this would roughly double the work.

\begin{lemma}
\label{lem:partition-size}
Every subproblem $A_j$ created in Step 1b has size between
$r^{c/2} + 1$ and $3 \cdot r^{c/2} - 1$, except possibly the last
subproblem, which may be smaller.
\end{lemma}

\begin{proof}
Let $\widetilde{S}$ be
the subset of $S$ comprising every $2r$-th item in $S$ used in Step 1b.
Let $e$ and $e'$ be two successive elements in sorted order in
$\widetilde{S}$.
We will argue that their ranks in the full set of $n$ items differ by
between $r^{c/2} + 1$ and $3 \cdot r^{c/2} - 1$.
Consider the $j$-th input list $L_j$ and let $SL_j$ be the sorted list
of samples in $S$ from $L_j$.
Suppose there are $k_j$ items from $SL_j$ in the range $[e,e']$.
Then there are between
$(k_j-1)\cdot r^{c/2-1} + 1$ and $(k_j+1) \cdot r^{c/2-1} -1$
items from $L_j$ in this range.
Summing over all $j$, as the sum of the $k_j$'s totals $2r$,
we obtain  that there are more than $r^{c/2} + r$ items
in this range and at most $3\cdot r^{c/2} - r$ items in this range.
The upper bound can be tightened by $2(r-1)$ on noting that the upper bound is not
met for the lists containing $e$ and $e'$.
\end{proof}

In the rest of the paper we give efficient implementations of SPMS
for cache-oblivious cache efficiency, for work-efficient parallelism,
for simultaneously achieving both of the above, and for
reducing the communication cost of false sharing.
In all of these implementations, Steps 2 and 3
remain unchanged, but we give several different implementations of
Step 1 with the goal of achieving efficiency with the most direct
method for the measure being considered.
This includes several different methods for
sorting the sample $S$ efficiently in Step 1a.

\subsection{A Cache Oblivious Sequential Implementation}\label{sec:co}

We consider a standard caching environment as in~\cite{FLPR99}: We
assume that data is organized in blocks of size $B$, and an ideal cache of size $M$; here,
an ideal cache is one that has an optimal cache replacement policy. We assume
the standard tall cache assumption,
namely that $M= \Omega(B^2)$, so the cache can hold $\Omega (B)$
blocks.
Since, with a tall cache,  $n/B$ dominates $\sqrt n$ when $n >M$, and when $n<M$,
a cache bound of $n/B$ implies that each data item is read into cache only once,
we have the following fact.

\begin{fact}
\label{fact:ideal-tall-cache}
With an ideal tall cache,
a cache miss bound of $n/B + \sqrt n$ can be bounded as $O(n/B)$.
\end{fact}

As usual, we will refer to a cache miss bound of $O(n/B)$ for accessing $n$ elements as the {\it scan bound} and a cache miss bound of $O(\frac{n}{B} \log_M n)$ as the {\it sort} bound; each of these two bounds is optimal for the corresponding problems of scanning and sorting.

\vone
\noindent
{\bf Data Layout.} Data layout is important for cache-efficiency. We
will assume that the $r$ input lists are stored in an array, starting
with the sorted  items in $L_1$, followed by $L_2$, and so
on. Similarly, we will store the sample in an array $S[1..s]$, which
will be a subsequence of the input array.

We now describe an efficient cache-oblivious implementation of Steps 1a and 1b.
For Step 1a, we will compute the rank of each element in $S$, which is essentially
equivalent to sorting $S$. We will not perform the final step of storing $S$ in sorted order,
so that we can avoid the cache misses that would result with this step.

\vhalf
\noindent
{\bf Step 1 of Cache-oblivious SPMS}

\vhalf
\noindent
{\bf Step 1a.} {\it Form the sample $S$ and rank each element of $S$  in $S$.}

$(i)$ For each of the $r$ sorted lists $L_j$, form the sublist $SL_j$
comprising
every $r^{(c/2)-1}$th item in $L_j$. The set $S$
comprises the
union of the items in the $SL_j$ lists.

$(ii)$ For each item in $S$, determine its rank in each $SL_j$, storing the ranks for each item contiguously as follows.

Let $|S|=s$.
Use an array of size $r \cdot s$, where the first $r$ locations will store the rank of $S[1]$ in each $SL_j$,
$1\leq j\leq r$, and  each subsequent group of $r$ locations
will store  the $r$ ranks of each subsequent element in $S$.
Scan $S$ $r$ times, once in turn for each of the lists
$SL_1, SL_2, \ldots, SL_r$;
the $j$-th scan,
which scans $S$ once and $SL_j$ $r$ times,
 records the rank in
$SL_j$ of each item $e$ in $S$.

$(iii)$ For each $e\in S$, sum its $r$ ranks computed in
Step 1a$(ii)$ to obtain its rank in set $S$.

\vhalf

\noindent
{\bf Step 1b.} {\it Form the sorted pivot set $P$ and partition the lists with respect to elements in $P$.}

$(i)$ Store the pivot set $P$ in sorted order by scanning the ranks of the items in $S$, and writing those items
with rank an integer multiple of $2r$
to their sorted location in an array of size
$s/(2r)$.
In general, contiguous inputs can be written to non-contiguous output locations
(and vice versa).

$(ii)$ Extract each subproblem $A_i$ in turn to complete the partitioning.
This will incur some caching penalty
over the scan bound because for each subproblem
fragments of $r$ input lists are read and written.

\vone
The following observation will be used in deriving our cache miss bound for the above steps.

\begin{observation}\label{obs:r>B}
If $r\geq B$ and $s = n/r^{(c/2)-1}$, then $r \cdot s = O(n/B)$ if $c\geq 6$.
\end{observation}

\begin{proof}
$r \cdot s = r \cdot n/r^{(c/2)-1}  = O(n/(B \cdot r^{(c/2)-3})) = O(n/B)$ when $c \geq 6$.
\end{proof}

\begin{lemma}
\label{lem:seq-cache-comp}
SPMS can be implemented to run in optimal sequential
cache-oblivious cache complexity
$O((n/B) \cdot \log_M n)$
if $c\ge 6$ and $M\ge B^2$.
\end{lemma}
\begin{proof}
Recall that $s= |S|$.

Step 1a$(i)$ is simply a scan of the full set of
items and so incurs $O(n/B)$ cache misses.
Step 1a$(iii)$ is a segmented sum computation
on an array of size $sr$ and so
can be performed in $O(sr/B)$ cache misses.
If $c\geq 4$, $s \cdot r = n/r^{(c/2)-2}$,
hence $s\cdot r/B= O(n/B)$ when $c\geq 4$.

Consider the reads in Step 1a$(ii)$. This step
performs $r$ scans over the $S$ array, and over
each of the $SL_j$ for the reads. This has caching cost $s \cdot r/B +
\sum_j r \cdot |SL_j|/B$.
As already noted, $s\cdot r/B= O(n/B)$ when $c\geq 4$.
To bound the caching cost for accessing the $SL_j$,
note that the total cost for the $r$ accesses to those
$SL_j$ that have at least $B$ items is $O(r \cdot s/B)$ since $s$ is
the sum of the sizes of all of the $SL_j$, and $| SL_j|/B \geq 1$ for
each of these lists.
For any $SL_j$ with fewer than $B$ items, we
observe that only one cache miss is needed to bring this list into
cache, and in an ideal cache, this list will stay in cache during the
entire scan of $S$ that determines the ranks of the elements of $S$ in
this $SL_j$,
as long as the number of such lists is less than $M/B$.
Hence the cost for accessing this class of lists is just
$r$ cache misses, which is $O(\sqrt n)$
when $c\geq 4$ since $n = \Omega(r^{c/2})$.
If $r>M/B$ then
the tall cache assumption implies
$r > B$, and by Observation~\ref{obs:r>B}
$r \cdot s = O(n/B)$
if $c \geq 6$.

Now we bound the writes in Step 1a$(ii)$. There are
$r \cdot s$ writes in this step, and each sequence of $s$
writes is written to positions $r$ apart in the output array. If
$r\geq B$,
$r \cdot s = O(n/B)$ if $c \geq 6$ by Observation~\ref{obs:r>B}.
If $r<B$, then this step writes $B/r$ successive
items with one cache miss, hence the cache miss cost is bounded by
$O(r\cdot s \cdot (r/B) )= O(n/(r^{c/2-3}B))$,
and this is $O(n/B)$ when $c\ge 6$.

In Step 1b$(i)$ the reads have cost $O(n/B)$ since the array being scanned has
size $O(n/r^{(c/2)-1})$.
Every $2r$-th of the scanned elements is written,
and as
each write could incur a cache miss, there are $O(n/r^{c/2})$ cache
misses for the writes. This is $O(\sqrt n)$ since $n \leq r^c$.

Finally, for Step 1b$(ii)$, as the subproblems are written in
sequence, we can bound the cache misses as follows.
If $r\le B$, it incurs only
$\sum_i |A_i|/B = O(n/B)$ cache misses
since $M\ge B^2$ and so the
(ideal) cache can hold the most recently accessed block from each of the
$r \leq B$ lists.
If $r\ge B$, it incurs
$\sum_i [|A_i|/B +O(r)]=n/B +\sum_i O(r)$ cache misses.
But if $r\ge B$,
$\sum_i O(r)=O(\frac{n}{r^{c/2}} \cdot r) = O(n/r^{c/2 - 1})= O(n/r) = O(n/B)$,
if $c \ge 4$.

Thus for Step 1, we have a cache miss bound of $O(n/B  + \sqrt n)$, which is $O(n/B)$
by Fact \ref{fact:ideal-tall-cache}.
This yields the following recurrence for $Q(r,n)$, the number of cache misses while
executing an instance of SPMS on an input of size
$n \leq r^c$
(recall that SPMS calls the multi-way merge with $n$ lists,
each containing one element, and with $n=r$ initially).
In this recurrence $n_{ij}$ is the size of the multi-way merge subproblem $A_{ij}$
in Step 2, and $n_i$ is the size of the $i$-th multi-way merge subproblem in Step 3.

\[
Q(n,r) = O(n/B) + \sum_{i,j} Q( n_{ij}, \sqrt r) + \sum_{i} Q(n_i, \sqrt r).
\]

The base case, $n \leq M$, has $Q(n,r) = O(n/B)$.
The solution to this recurrence is the desired optimal caching bound
$Q(n,r) = O\left(\frac{n}{B} \frac{\log n}{\log M}\right)$.
\end{proof}

\noindent
{\bf Comment}. The bound in Lemma~\ref{lem:seq-cache-comp} can be generalized to handle
a smaller tall cache assumption, namely an assumption that $M \ge B^{1+\eps}$
for any given constant $\eps >0$, as follows.
We will now need that
$c \ge 2 + 2 \cdot \frac{1 + \eps} {\eps}$.
The analysis of Step 1 now splits into cases according as $r$ is smaller or larger
than $B^{\eps}$.
Observation~\ref{obs:r>B}, which is used
in the analysis of Step 1a(ii), now holds if $r \ge B^{\eps}$;
also
in the analysis of Step 
1b(ii), we obtain that when $r \ge B^{\eps}$,
$\sum_i r = O\left(\frac{n}{r^{c/2}} \cdot r \right)= O(n/r^{\frac{c}{2}-1})= O(n/r^{1/\eps}) = O(n/B)$.

\subsection{A Parallel Implementation}

Here we follow the framework of the cache-oblivious algorithm in the previous section.
The implementation is changed in a few places as follows.

In Step 1a$(ii)$ we perform $r$ binary searches for each item $e$ in $SL$
to determine its rank in each of the lists $SL_j$.
To carry out Step 1b$(ii)$,
for each item in the input, we will compute its destination location in the partitioning;
then it takes a further $O(1)$ parallel time to perform the partitioning.
To this end, we find the rank of each of the $n/(2\cdot r^{c/2})$ pivot items in each list $L_j$
(by means of $n/(2\cdot r^{c/2-1})$ binary searches)
and we also find for each pivot item $e$,
the destination rank of its successor in each list $L_j$.
These can be computed using
a prefix sums computation
following the binary searches.
Then a parallel scan across the lists $L_j$ yields the destination locations.
The rest of the computation in Steps 1a and 1b is readily performed in logarithmic time and linear work.
This leads to the following lemma.

\begin{lemma}
\label{lem:merge-anal}
There is an implementation of SPMS that uses
$O(n \log n)$ operations
and $O(\log n \log\log n)$ parallel steps
on an input of size $n$, if $c \ge 6$.
\end{lemma}
\begin{proof}
Clearly, the work, i.e., the cost or operation count,
and parallel time for Step 1, apart from the binary searches,
are $O(n)$ and $O(\log n)$ respectively.
Each binary search takes $O(\log r)$ time and there are
$O(r\cdot n/r^{c/2-1})$ binary searches,
giving an overall work of $O(n/r^{c/2-2})$ for the binary searches;
this is an $O(n)$ cost if $c \ge 4$.

Let $W(r,n)$ be the operation count for a collection of
multi-way merging problems of total size $n$,
where each comprises the merge of at most
$r$ lists of combined size at most $r^c$.
Then we have:
$ W(r,n) \leq O(n) + 2W(r^{1/2},n) = O(n \log r)
= O(n \log n)$.

For the parallel time,
we have seen that Step 1 runs in $O(\log n)$ time.
Hence, the parallel run time $T(r,n)$ is given by:
$T(r,n) \leq O(\log n) + 2 T (\sqrt r, r^{c/2}) = O( \log n \log\log r) = O( \log n \log\log n)$.
\end{proof}

Lemmas \ref{lem:seq-cache-comp}
and \ref{lem:merge-anal} lead to the following corollary.

\begin{corollary}
\label{cor:sort-time}
$~$
\\
i.
If $M= \Omega (B^{1+\epsilon})$, for any given constant $\epsilon > 0$,
SPMS can be implemented to sort $n$ elements
cache-obliviously in $O(n\log n)$ time with $O((n/B) \log_M n)$ cache
misses.
\\
ii. SPMS can be implemented to sort $n$ elements in
$O(\log n \cdot \log\log n)$ parallel steps and $O(n\log n)$ work.
\end{corollary}

In Section~\ref{sec:co-parallel}
we will describe an implementation
of Step 1 that will give rise to a
parallel, cache-oblivious version of SPMS that
simultaneously attains the above work, parallel time, and cache-oblivious
cache-efficiency
(when considering only the standard caching costs considered in
sequential analysis).
Then, in the following sections, we describe and bound the effects of `false-sharing'.

\section{Multicore Computation Model}\label{sec:model}

We consider a multicore environment consisting of
$p$ cores (or processors), each with a private cache of size $M$.
The $p$ processors communicate through an arbitrarily large shared memory.
Data is organized in blocks (or `cache lines') of size $B$.

If we ignore communication costs a multicore is simply an asynchronous PRAM~\cite{KRS90}.
In contrast to a completely asynchronous environment, in this paper
we use a computation model in which synchronization is performed using binary forks and joins;
we describe this model in Section~\ref{sec:c-dag}.
Parallel tasks are assigned to idle processors by a scheduler, and we
discuss the type of schedulers we consider in Section~\ref{sec:sched}.
The communication costs in the multicore arise from the delay when transferring a requested piece of data
from shared memory to the cache of the processor requesting the data. We discuss the caching costs in Section~\ref{sec:cache-cost}.
We discuss the notion of resource obliviousness in Section~\ref{sec:ro}. Finally, in Section~\ref{sec:model-related} we
discuss related models and schedulers as well as the connections between this current paper and its
earlier proceedings version in~\cite{CR10b}.

\subsection{Computation Dag Model}\label{sec:c-dag}

We will  express parallelism through paired fork and join operations.
A fork spawns two tasks that can execute in parallel. Its corresponding join is a
synchronization point: both of the spawned tasks must complete before the computation
can proceed beyond this join.

The building block
for our algorithms, which we call a \emph{fork-join computation}, comprises
a height $O(\log n)$ binary tree of fork nodes with possibly $O(1)$ additional computation at each node,
with $n$ leaf nodes,
each performing a sequential
computation of length $O(1)$ or $O(\log n)$,
followed by a complementary tree of join nodes on the same leaves,
again with possibly $O(1)$ additional computation at each node.
In general, one might have other lengths of sequential computation, but they do not occur
in the SPMS algorithm.

A simple example of a fork-join computation is the computation of the sum of the
entries in an array, by means of a balanced tree, with one leaf node for each
array entry, and each internal node returning to its parent
the sum of the values computed by its children.
Good cache performance occurs if, as is natural, the left-to-right
ordering of the leaves matches the order of the entries in the array.
Prefix sums can be computed by sequencing two fork-join computations.

The overall SPMS algorithm is created by means of sequencing
and recursion, with the fork-join computations forming the non-recursive portions
of the computation. Such a computation contains only nested fork-join computations, and it
generates a series parallel graph. The {\it computation dag} for a computation on a given input is
the acyclic graph that results when we have a vertex (or node)
for each unit (or constant) time computation, and
a directed edge from a vertex $u$ to a vertex $v$ if vertex $v$ can begin its computation only after
$u$ completes its computation. For an overview of this model, see Chapter 27 in~\cite{CLRS09}.

\subsection{Scheduling Parallel Tasks}\label{sec:sched}

During the execution of a computation dag  on a given input, a parallel task is created each time a
fork step $f$  is executed.
At this point the main computation proceeds with the left child of $f$ (as in a standard
sequential dfs computation), while  the task $\tau$ at the right child $r$
of the fork node is a task available
to be scheduled in parallel with the main computation.
This parallel task $\tau$ consists of all of the computation
starting at $r$ and ending at the step before the join corresponding to the fork step $f$.

Consider the execution on $p$ processors of a computation dag $D$ on a given input.
In general, the parallel tasks generated during this execution can be scheduled across the $p$ processors in
many different ways, depending on the policy of the scheduler used to schedule parallel tasks.
For example,
the task $\tau$ mentioned in the above paragraph could be moved by the scheduler
to a processor $Q$ different from $P$, the processor executing the main computation,
in which case $\tau$ will execute on $Q$ in parallel
with the execution of the main computation on $P$.
The scheduler could also choose not to schedule $\tau$
on another processor (perhaps because all other processors are already executing other computations),
and in this case, $\tau$ will be executed by $P$ according to the sequential execution order.
In the case when the parallel task $\tau$ is moved from $P$ to
another processor $Q$,
we will say that processor $Q$ {\it steals}
$\tau$ from $P$; this terminology is taken from the class of work-stealing schedulers,
which is the class of schedulers we consider in this paper.

\paragraph{Work-stealing Schedulers.}
Under work-stealing, each processor maintains a task queue on which it enqueues the parallel tasks it generates.
When a processor is idle it attempts to obtain a parallel task from
the head of the task queue of another processor
that has generated parallel tasks.
The exact method of identifying
the processor from which to obtain
an available parallel task determines
the type of work-stealing scheduler being used.
The most popular type is randomized work-stealing (RWS, see e.g.,~\cite{BL99}), where a
processor picks a random processor and steals
the task at the head of its task queue, if there is one.
Otherwise, it continues to pick random processors and tries to find an available parallel task
until it succeeds, or the computation completes.

\subsection{Execution Stacks}
\label{sec:exec-stack}

Let us consider how one stores variables that are generated during the execution of the
algorithm.
It is natural for the original task and each stolen task to each have an
execution stack on which they store the variables declared by their residual task.
$E_{\tau}$ will denote the execution stack for a task $\tau$.
As is standard, each procedure and each fork node stores the variables
it declares in a \emph{segment} on the current top of the execution stack for the
task to which it belongs,
following the usual mode for a procedural language.
As implied by its name, an execution stack is accessed by a processor in stack order.

\paragraph{\bf Execution Stack and Task Queue.}
The parallel tasks for a processor $P$ are enqueued on its task queue in the order in which their
associated segments are created on $P$'s execution stack. The task queue is a double-end queue,
and $P$ will access it in stack order similar to its accesses to its execution stack. 
Thus, $P$ will remove an enqueued task $\sigma$ from its task queue when it begins computing on
$\sigma$'s segment.

Under work-stealing, the task that is stolen (i.e.,
transferred to another processor) is always the one that is at the top of the task queue in the
processor from which it is stolen, or
equivalently, tasks are stolen from a task queue in queue order.

\paragraph{\bf Usurpations of an Execution Stack.}
Suppose that  a subtask $\tau'$ is stolen from a task $\tau$.
Consider the join node $v$
immediately following the node
at which the computation of $\tau'$ terminates.
Let $C$ be the processor executing $\tau - \tau'$ when it reaches node $v$ and
let $C'$ be the processor executing $\tau'$ at this point.
To avoid unnecessary waiting, whichever processor (of $C$ and $C'$) reaches $v$ second
is the one that continues executing the remainder of $\tau$.
If this processor is $C'$, we say that $C'$ has {\it usurped}
 the computation of $\tau$.
We will examine the additional false sharing costs incurred due to usurpations in
Section~\ref{sec:fs-estack}. There are additional regular caching costs as well, but these are
handled by the analysis described below in Section~\ref{sec:cache-ws}.

\subsection{Caching Cost}\label{sec:cache-cost}

As in a sequential computation, a processor sustains a cache miss when accessing a data item if
that item is not in its cache. When we bound the cache miss cost of a parallel execution,
we will assume
an optimal cache replacement policy (as in the sequential case).
The caching cost for a parallel execution
is the sum of the caching costs at each of the $p$ parallel cores.
There is a second source of cache misses
that is present only in the parallel setting: {\it false sharing}.
We will discuss and analyze this cost in
Sections~\ref{sec:fs-misses}--\ref{sec:fs-estack}.
Here we confine our attention to regular cache misses.

\subsubsection{Execution Stack and Cache Misses due to Misalignment}

A parallel execution can incur additional cache misses due to
{\it block misalignment.}
Block misalignment can arise whenever there is a new execution stack, for the block boundaries
on the new stack may not match those in the sequential execution.
In fact, block misalignment, which was overlooked previously,
arises even in analyses in the style of~\cite{ABB02}.
But, as we will see,
in Corollary~\ref{cor:block-misalign},
the cost due to block misalignment is at most the cost
increase that occurs if the cache sizes are reduced by a constant factor, and for
{\it regular algorithms}~\cite{FLPR99},
which cover most standard algorithms including SPMS, this
increases the cache miss costs by only a constant factor.

Until now our analysis had assumed that the parallel execution created the same
blocks as the sequential execution and argued that the resulting cache miss bound was of the form
$O(Q + \mbox{additional terms})$, where $Q$ was the cache miss cost of a sequential
execution.
But, in fact, this assumption need not hold.

Let us consider what happens when a subtask $\sigma$ is stolen
 from task $\tau$.
The variables that $\sigma$ declares are placed on its execution stack,
presumably starting at the boundary of a new block.
However, if there had been no steal, these variables would have been placed
on the continuation of $\tau$'s execution stack, presumably, in general, not starting
at a block boundary. This misalignment may cause additional cache misses compared
to the sequential execution.

More specifically, any block $\beta$ on the execution stack in the sequential algorithm may
be apportioned to multiple execution stacks in a parallel execution, with one of these
stacks potentially using portions of two of its blocks to hold its portion of $\beta$.

To bound the cost of cache misses, it will suffice to bound, for the processor $P_{\sigma}$
executing $\sigma$, for each block $\beta$ in the sequential execution
storing data that $P_{\sigma}$ needs
to access,
the number of blocks storing portions of $\beta$ that $P_{\sigma}$ accesses
in the parallel execution.

\begin{lemma}\label{lem:misalign-cost}
Let $P_{\sigma}$ be a processor executing a stolen subtask $\sigma$.
Suppose $P_{\sigma}$ accesses data which is stored in block $\beta$ in the
sequential execution.
Then $P_{\sigma}$ needs to access at most 4 blocks in the current parallel execution
to access this data.
\end{lemma}
\begin{proof}
To analyze this,
let us consider which variables $P_{\sigma}$ can access.
$P_{\sigma}$ can access the $O(1)$ variables declared by the node from which
$\sigma$ was forked.
It can also access variables declared by its
``calling'' SPMS task (the SPMS task that owns the steal that created $\sigma$).
Finally, it can access the variables stored on $\sigma$'s execution stack.
These variables could be on as many as three different execution stacks,
one for the calling task, one for the parent fork node,
and one for $\sigma$.

Consider a block $\beta$ in the sequential execution, some of whose contents
$P_{\sigma}$ seeks to access.
If $\beta$'s contents are stored in a single execution stack in the current parallel execution,
then $\beta$'s contents occupy at most two blocks in the current parallel execution.
On the other hand, if in the current parallel execution,
$\beta$'s contents are stored over two of the execution stacks that $P_{\sigma}$ accesses,
then $\beta$'s contents occupy at most three blocks that $P_{\sigma}$ accesses
(this can happen for at most two distinct blocks $\beta$ accessed by $P_{\sigma}$).
Finally, if $\beta$'s contents are stored over three execution stacks that $P_{\sigma}$ accesses,
then $\beta$'s contents occupy at most four blocks that $P_{\sigma}$ accesses
(this can happen only once for data being accessed by $P_{\sigma}$,
and only if the ``over two execution stacks'' scenario does not occur).

It follows that for each block $\beta$ accessed in the sequential execution of $\sigma$,
$P_{\sigma}$ will need to access at most 4 blocks.
\end{proof}
\begin{corollary}\label{cor:block-misalign}
The block misalignment
in the parallel execution of SPMS increases the cache miss bound by just a constant factor.
\end{corollary}
\begin{proof}
By Lemma~\ref{lem:misalign-cost},
if the sequential execution uses a cache of size $M$
and the parallel execution uses a cache of size $4M$,
then the parallel execution will incur at most 4 times as many cache misses.

By regularity~\cite{FLPR99},
as the cache miss bounds just use the assumption that
$M\ge B^2$ and are polynomial in $B$,
if the parallel execution actually had a cache of size $M$, it would increase
the cache miss bound by just a constant factor.
\end{proof}

\subsubsection{Caching Costs Under Work Stealing}\label{sec:cache-ws}

Suppose a parallel execution incurs $S$ steals
and let $R(S)$ be a bound on the number of additional cache misses this computation incurs
as a function of $S$.
There is a simple upper bound on $R(S)$ when using a work-stealing scheduler.
Consider the sequence of steps $\sigma$ executed in a sequential execution.
Now consider a parallel execution that
incurs $S$ steals.
Partition $\sigma$ into contiguous portions so that in the parallel execution
each portion is executed in its sequential order on a single processor.
Then, each processor can regain the state of the
cache in the sequential execution once it has accessed
$O(M/B)$ distinct blocks during its execution (see, e.g.,~\cite{ABB02}
in conjunction with Lemma~\ref{lem:misalign-cost}).
Thus if there are $K$ portions, then there will be $R(S) = O(K\cdot M/B)$ additional cache misses.
Furthermore, it is shown in~\cite{ABB02} that $K \le 2S + 1$ for a work-stealing scheduler.
We review this bound
in Lemma~\ref{lem:kernel-bound},
using the notion of \emph{task kernels}.

The following fact is well-known (see, e.g., \cite{CR12}).

\vone
\noindent
{\bf The Steal Path Fact (for Work-stealing Schedulers).}
Let $\tau$ be either the original task or a stolen subtask.
Suppose that $\tau$
incurs steals of subtasks $\tau_1, \cdots , \tau_k$.
Then there exists a path $P_{\tau}$ in $\tau$'s computation dag from its
root to its final node such that the parent of every stolen
task $\tau_i$ lies on $P_{\tau}$, and every off-path right child of
a fork node on
$P$ is the start node for a stolen subtask.

\noindent
{\it Comment.}
The above fact follows by observing the nature of a depth-first-search
based sequential execution of recursive computations,
and the dequeue mode of steals under work-stealing.

\paragraph{\bf Tasks and Task Kernels.}
Given a parallel execution that incurs steals, the following definition
gives a partition of the computation dag into a collection of
{\it task kernels} induced by the steals.

\begin{definition}
\label{def:task-kernel-work-stealing}
{\bf  (Task kernels under work-stealing.)}
Consider a parallel execution of SPMS under work-stealing, and let it incur $S$ steals,
numbered as $\sigma_1, \sigma_2 , \cdots , \sigma_S$ in the
order they occur relative to a sequential execution.
In turn,
we partition the computation dag
into task kernels with respect to the sequence
$\Sigma_i  = \langle \sigma_1, \sigma_2, \ldots, \sigma_i \rangle$
to create the collection $C_i$.
We let $\Sigma_0$ be the empty sequence and its associated partition $C_0$ be the
single task kernel containing the entire computation dag.
For each $i \geq 1$ the partition $C_{i+1}$ is obtained from $C_i$
as follows.
Let $\tau$ be the task kernel in $C_i$
that contains the fork node $v_f$  at which steal $\sigma_{i+1}$ is performed,
and let $v_j$ be the corresponding join node.
Then, $\tau$ is partitioned into the following three task kernels.

\begin{enumerate}
\item
$\mu_1$, the stolen subtask created by $\sigma_{i+1}$.
\item
$\mu_2$, the portion of $\tau$ preceding the stolen subtask in the sequential execution.
(This includes the portion of the computation  descending from  the left child of $v_f$
that precedes $v_j$.)
\item
$\mu_3$, the  portion of $\tau$ descendant from $v_j$ in the computation dag, including
node $v_j$ itself.
This is the remainder of $\tau$.
\end{enumerate}

Then, $C_{i+1} = C_i - \{\tau\} \cup \{ \mu_1, \mu_2, \mu_3\}$.

For this parallel execution of SPMS, the collection of task kernels is $C_S$.
\end{definition}

\begin{lemma}
\label{lem:kernel-bound}
Consider an execution of SPMS
under a work stealing scheduler,
and suppose it incurs $S$ steals.
The resulting partition into task kernels forms at most
$2S+1$ kernels.
Furthermore, each task kernel is executed on a single processor and each kernel
forms a contiguous portion of work in the sequential execution.
\end{lemma}
\begin{proof}
If $S=0$, the entire computation is contained in one task kernel,
and $1=2 \cdot 0 +1$.
Each additional
steal partitions an existing task kernel into three new task kernels, and hence the size of the collection increases by two as needed.

To see that each task kernel is executed by a single processor, it suffices to
note that the only edges along which the executing processor can change
are the edges from a fork node into the initial node of a stolen subtask, and
the edges into the join node immediately following the end of a stolen subtask;
this describes precisely where the partitioning into kernels occurs.

Finally, the fact that a kernel forms a contiguous portion in the sequential execution
follows from the dfs-based sequencing of the computation dag in the sequential execution.
\end{proof}

\begin{corollary}
\label{cor:simple-rws-miss-bound}
Consider a parallel execution of SPMS using a work-stealing scheduler.
Suppose it incurs $S$ steals.
Then $R(S)$, the additional cache misses it incurs, above the cost of the
sequential execution, is bounded by $R(S) = O(S \cdot M/B)$.
\end{corollary}

\subsubsection{Scheduler Cache Miss Performance}

Consider a computation with span $T_{\infty}$, whose sequential running time is
$T_1$, and sequential cache miss bound is $Q$.
Suppose a parallel execution incurs $S$ steals and the cost of $F(S)$ cache misses due to
false sharing (which will be discussed in Sections~\ref{sec:fs-misses}--\ref{sec:fs-estack}).
Let $T_s$ be the total time spent across all processors on performing successful steals
 and let $T_u$ be
the total time spent on unsuccessful steals. 
Let $b$ be the cost of a cache miss and $s$ the cost of a successful steal.
Finally, let $I$ be the total time spent by idling processors. Since every processor is computing, or waiting on a cache miss or false sharing miss,
or attempting a steal, successful or otherwise, or simply idling, we can bound the execution
time $T_p$ on $p$ processors as the following. (The steal cost $s$ does not appear in the
following equation since it is subsumed into the cost for steals, but it will be used in
Section~\ref{sec:perf} when we apply this bound to randomized work stealing.)

\begin{equation}
\label{eqn:rws-time}
T_p = \frac{1}{p} \left( T_1 + b \cdot Q + b\cdot S \cdot \frac{M}{B} + + b \cdot F(S) + T_s  +  T_u + I \right).
\end{equation}

Equation~\eqref{eqn:rws-time} does not incorporate
the overhead due to false sharing, except as an additive term $F(S)$.
This term will be discussed in
Sections~\ref{sec:fs-misses} and~\ref{sec:fs-estack}.

As noted earlier, randomized work stealing (RWS) is the most popular form of work-stealing, and for this method
good bounds are known for the expected number of steals~\cite{BL99,ABB02,CR13}. We will discuss this in Section~\ref{sec:perf}, after
we have discussed issues relating to false sharing.

\subsection{Resource Obliviousness}\label{sec:ro}

We use the notion of resource obliviousness described in~\cite{CR12}. A multicore algorithm is
resource oblivious if the algorithm does not use any machine parameters such as $p$, $M$, and $B$,  and
yet is analyzed to run efficiently
in a parallel execution under a range of parameters.
As noted in~\cite{CR12}
resource obliviousness differs from the notion of multicore-oblivious introduced in~\cite{CRSB13}
by being scheduler independent.

We obtain bounds for our resource oblivious algorithms
in terms of the machine parameters $p$, $M$, and $B$, the input size $n$, and the number of
steals $S$; $S$ in turn may depend on the scheduler that handles
the task transfers,
but this is the only dependence of our algorithm and its analysis on the scheduling.

For SPMS, we have
presented an algorithm with an optimal bound for the
sequential cache complexity $Q(n)$, and
another algorithm that performs $O(n \log n)$ operations and achieves
span $T_{\infty}= O(\log n\log \log n)$.
In Section~\ref{sec:co-parallel} we describe an SPMS implementation that
simultaneously achieves both of these bounds. Thus when executed by a work-stealing scheduler
Equation~\eqref{eqn:rws-time} will apply.
This is a resource oblivious result, except that the analysis does not consider the cost of
false sharing.

In Sections~\ref{sec:fs-misses}--\ref{sec:fs-estack} we discuss the cache miss overhead due
to false sharing and present the final version of our SPMS algorithm.
Its cost overhead for false sharing is
the cost of $O(S \cdot B)$ cache misses.
 With a tall cache we have $M \geq B^2$, thus this
false sharing cost is dominated by the cache miss overhead of $S \cdot M/B$ shown
in Corollary~\ref{cor:simple-rws-miss-bound}.

\subsection{Related Work}\label{sec:model-related}

As mentioned earlier, if we ignore communication costs, then the multicore is an asynchronous PRAM.
When we incorporate communication costs, our model differs from the BSP~\cite{V90}, LogP~\cite{CK+93},
QSM~\cite{GMR97}, and the multi-BSP~\cite{Va08} in two ways: it uses cache misses instead of
gap and/or latency parameters for the communication cost, and it uses fork-join parallelism instead
of  bulk-synchrony.

The PEM model~\cite{AGN08} uses a parallel machine similar to ours: a shared memory model, where
each processor has a private cache, and communication cost is measured in cache misses. However,
the PEM model again uses bulk-synchrony. Further false sharing is not addressed in results
known for PEM, nor is the notion of resource obliviousness.

Orthogonal to our work are results on `multicore-oblivious' algorithms for a multi-level caching hierarchy in \cite{CRSB13}
(see also~\cite{BFGS11,SBFGK14}).
These results deal with a specific scheduler, and while the algorithms do not use any multicore parameters,
the scheduler is aware of these parameters.
This scheduler dependence on machine parameters appears to be
necessary when considering a multi-level tree of caches. Our results deal with private caches only;
as a result we are able to establish  resource oblivious bounds for SPMS that are
independent of the scheduler used.

\vone
\noindent
{\bf Relation to \cite{CR10b}}.
The SPMS algorithm was first presented in the extended abstract~\cite{CR10b}, and the basic algorithm is the same both here and
in~\cite{CR10b}. However, there are some differences between the results presented in~\cite{CR10b} and the results we present in
this paper, and we list the key differences:

\begin{enumerate}
\item False sharing was considered in~\cite{CR10b}, but the results there applied only to the case
when space is not deallocated (as discussed in Section~\ref{sec:fs-misses}). Here, we additionally
consider space deallocation on execution stacks, and the resulting false sharing costs (in
Section~\ref{sec:fs-estack}).

\item The SPMS implementation in~\cite{CR10b} was
analyzed for two specific work-stealing schedulers: the
Priority WorkStealing (PWS)
scheduler~\cite{CR10a}, which is a
deterministic variant of a work-stealing scheduler,
and RWS.
The PWS scheduler requires additional properties of balance in forked tasks (using
BP and HBP computations~\cite{CR10a,CR12}). While these results are interesting on their own, they are
orthogonal to the enhanced
scheduler-independent
resource oblivious notion in~\cite{CR12} which we use here.
Hence the results in~\cite{CR10b} relating specifically to PWS are not included here.
Instead, here we analyze more generally in terms of work stealing schedulers.
\end{enumerate}

\section{Cache-Oblivious Parallel SPMS}
\label{sec:co-parallel}

We would like to simultaneously achieve the
$O(\frac nB \log_M n)$ cache miss bound of Lemma~\ref{lem:seq-cache-comp} and the
$O(\log n \log\log n)$
parallel time of Lemma~\ref{lem:merge-anal}.
To achieve this, it suffices to implement Step 1 so that it runs in parallel time $O(\log n)$
on an input of size $n$, and incurs $O(n/B)$ cache misses in a sequential execution.
This follows because
the forks and joins needed to instantiate the recursive calls in Steps 2 and 3
can be performed with height $O(\log r) = O(\log n)$ trees.

We describe our parallel cache-oblivious implementation of Step 1
and its analysis in Section~\ref{sec:par-co-step1}.
It uses four procedures,
of which the \emph{Transposing Redistribution}
may be of broader interest.
The other  procedures are named
\emph{TR Prep} (short for Transposing Redistribution Preparation),
\emph{Permuting Write}, and
{\it Small Multi Merge}.
We detail these four procedures in Section~\ref{sec:4procs}.

\subsection{Four Procedures for Step 1}\label{sec:4procs}

\noindent
1. {\bf  Transposing Redistribution}.

\begin{itemize}
\item[{\it Input.}]
A vector $Y$ which we view as consisting of the sequence
$Y_{11}, Y_{12}, \cdots, Y_{1k}$, $\cdots$, $Y_{r1}, \cdots Y_{rk}$ of subvectors.
These subvectors are specified with a list of the subvector
lengths $|Y_{ij}|$ and  their start indices $s_{ij}$ in the input vector, but
ordered in column major order,
which is the desired order in the output.

In  other words, the input is
$(Y; L)$, where $Y$ is the input vector and $L$ is
the ordered list
$(|Y_{11}|,s_{11}), (|Y_{21}|,s_{r1}), \cdots, (|Y_{r1}|,s_{r1}), \cdots,$
$ (|Y_{1k}|,s_{1k}),
\cdots, (|Y_{rk}|,s_{rk})$.

\item[{\it Output.}]
The transposed sequence $ Y_{11}, Y_{21}, \cdots, Y_{r1}, \cdots$,
$Y_{1k}, \cdots, Y_{rk}$.

\item[{\it Bound.}] The output is computed with $O(|Y|/B + r^2 k/B)$ cache misses.

\item[{\it Method.}]
First, by means of a prefix sums computation over the values $|Y_{ij}|$, we
determine the output location for each subvector $Y_{ij}$.
Then we copy these vectors to their output locations,
using a nested pair of loops, the outer loop being
a fork-join computation
over the subvectors in their output order,
and the inner loop being
a fork-join computation
over the individual elements in each vector.
The data in the second part of the input is exactly what is needed for
this computation:
namely for each subvector, where it occurs in the input, and what is its size.

\item[{\it Analysis.}]
Since the elements are accessed in scan order within each of the $rk$ subvectors, the
cache miss cost is $|Y|/B + rk$.
If $r \le B$, assuming $M \ge B^2$, the most
recently accessed block
$Y_{i j_i}$, for each $i$, will fit in cache, and then the cost reduces to $|Y|/B$.
While if $r > B$, $rk \le r^2 k/B$,
so we can restate the cost as $O(|Y|/B + r^2 k/B)$.
This gives the desired bound.
Note that we would have achieved a similar cache miss bound (interchanging $r$ and $k$)
if we had used a fork-join computation
for the outer loop with the subvectors in input order.
We use the output order because that also
gives good performance in the presence of false sharing,

\item[{\it Notes.}]

$(i)$  The second part of the input is in column major order so as to facilitate the creation of the output.\\
$(ii)$  Although our algorithm for Transposing Redistribution works with any input of the type described above,
we will call it with sorted lists partitioned using a single sequence of pivots for all $r$ lists.
For this class of inputs, the procedure TR Prep, given below, generates the input for Transposing Redistribution in the format described above, when it is given as input the $r$ sorted lists, and the list of pivots in sorted order.

\end{itemize}

\vhalf
\noindent
2. {\bf TR Prep}.
This will be used to set up the input for Transposing Redistribution.
\begin{itemize}
\item[{\it Input.}]
A sequence $Y$ of $r$ sorted lists $Y_1, Y_2, \cdots Y_r$, and
a sorted subset
$P = \{p_1,p_2, \ldots, p_{k-1}\} \subset Y$ of pivots.
It is convenient to add the dummy items $p_0 = -\infty$ and $p_k = +\infty$ to $P$.
In addition, each
item $e$ in $P$,
other than $p_0$ and $p_k$,
will have a pair of straddling items
in each list $Y_i$.
Our cache miss bound will depend on the parameter $d$,
where the straddling items are at most $d$ positions apart in
every $Y_i$.
In our CO-Parallel-SPMS algorithm (given in Section~\ref{sec:par-co-step1}),
$d$ will take on two different values: $r$ in Step 1a(ii)I,
and $r^{c/2-1}$ in Step 1a(ii)III.
The sorted pivot sequence $P$, together with the straddling items, are generated in CO-Parallel-SPMS
using the fourth procedure given below (Small Mulit-Merge).

\item[{\it Output.}]
Let $s_{ij}$ be the rank in $Y_i$ of the smallest item greater than or equal to $p_{j-1}$,
which for brevity we call $p_{j-1}$'s rank in $Y_i$.
TR Prep in effect partitions $Y_i$ into the sequence of sorted sublists
$Y_{i1}, Y_{i2}, \ldots, Y_{ik}$, where the items in $Y_{ij}$ lie in the range
$[p_{j-1},p_j)$.
This is done by computing the following output:
$(|Y_{11}|, s_{11}), (|Y_{21}|, s_{21}), \ldots, (|Y_{r1}|, s_{r1}), \ldots,
(|Y_{1k}|, s_{1k}), \ldots, (|Y_{rk}|, s_{rk})$.
Together with the list $Y$ this provides the input for Transposing Redistribution.
\item[{\it Bound.}] Our method for TR Prep incurs $O(kdr/B + kr^2/B + y/B)$ caches misses, where $y=|Y|$.

\item[{\it Method.}]
~~~$(i)$ For each $p_j$, $1\leq j \leq k-1$, perform a binary search in the interval spanned by its
straddling items in each of the $r$ lists in order to find $p_j$'s rank in each these lists. These
are the values for $s_{i,j+1}$, $1\leq i \leq r$, computed in the desired output
order $s_{11}, s_{21}, \cdots, s_{r1}, \cdots , s_{1,k}, \cdots s_{rk}$.

$(ii)$ For each $p_j$, for each $Y_i$, compute the length of the sublist of $Y_i$ straddled by $p_j$ and
$p_{j+1}$; this is computed by means of two lockstep scans of the results from $(i)$, comparing
the ranks for successive items in $P$. This give the $r$-tuple of values $|Y_{i,j+1}|, 1 \leq i \leq r$.

$(iii)$ Since both the $s_{ij}$ and $|Y_{ij}|$ values have been computed in the desired output order in
steps $(i)$ and $(ii)$, the output can now be written in a lock-step scan of these two lists.

\item[{\it Analysis.}]
Let $y = |Y|$, and recall that $|P|= k-1$.

Step $(i)$ incurs $O(kr\log \ceil{d/B}) = O(kr \cdot d /B + kr)$ cache misses.
If $d \ge B$ this is $O(krd/B)$ cache misses.
If $d < B$ and $r \ge B$ this is $O(kr^2/B)$ cache misses.
If $d < B$ and $r \le B$, with an ideal cache  the most recently searched block
for each list can be kept in cache, and as $P$ is in sorted order,
the cache miss cost reduces to $y/B$.
Steps $(ii)$ and $(iii)$  incur the scan bound of  $O(r \cdot k/B)$ cache misses.

Thus, this is always bounded by $O(kdr/B + kr^2/B + y/B)$ cache misses.
\end{itemize}

 \vhalf
\noindent
3. {\bf Permuting Writes}.

\begin{itemize}
\item[{\it Input.}]
A permutation array $P$  of size $x$, and an input array $A$  array of size $\ell x$, for some integer $\ell \geq 1$, where only elements at positions $\ell \cdot i$, $1\leq i \leq x$, are relevant.

\item[{\it Output.}] An output array $C$ of size $\ell' x$, with $C[\ell' \cdot P[i]]$ containing the element
$A[\ell i]$, for $1\leq i \leq x$.
In other words,
every $\ell$-th item in $A$ is copied to every $\ell'$-th location in
array $C$ in the permuted order given by $P$.
As with the input, only elements at positions $\ell' \cdot i$, $1\leq i \leq x$, in $C$ are relevant.

\item[{\it Bound.}] This is computed with
$ O(\ell x/B  + \ell' x/B + x^2/B)$
 cache misses.

\noindent
\item[{\it Note.}] If $\ell = \ell'= 1$, then every position of arrays $A$ and $C$ is relevant, and $C$ needs to contain the elements in $A$ permuted according to $P$.
We will use $\ell=1$, $\ell'>1$ in our algorithm for
the Small Multi Merge procedure, defined below.

\item[{\it Method.}]
We compute the output of the permuting writes in an array $C'$ of size $x$
with a fork-join computation where the leaves correspond to the elements in vector
$P$ in their input order.
Then, with a scan, we spread the output elements out to positions $\ell' \cdot i$ in the output array $C$.

\item[{\it Analysis.}]
Reading the input array $A$ incurs $O(\ell x/B)$ cache misses.
The writes into array $C'$ could incur up to $x$ cache misses,
while spreading the elements in the output locations in array $C$ incurs $O(\ell' x/B)$ cache misses.
Thus the cache bound is $O(\ell x/B  + \ell' x/B )$ plus the cache miss cost $Z$ for the permuting writes into
array $C'$.

If  array $C'$
fits in cache, i.e.\ if $x \le B^2$,
then each block in array $C'$ is read into cache only once and then $Z = O(x/B)$.
Otherwise, $x > B^2$,
hence $ \sqrt{ x}/B > 1$,
and then we obtain $x \le x^{3/2}/B$.
Thus the cache miss cost for permuting writes is bounded by
$O(\ell x/B  + \ell' x/B + x^{3/2}/B) = O(\ell x/B  + \ell' x/B + x^2/B)$ cache misses.

\end{itemize}

\vhalf
\noindent
4. {\bf Small Multi Merge.}

\begin{itemize}
\item[{\it Input.}]
The input $P = <W_1, \cdots , W_h>$ is an ordered sequence of
multi-way merging problems
(similar to the input to Step 3 of SPMS in Section~\ref{sec:intro}).
Each multi-way merging problem will contain $r$ lists of total length at most $x$, and
across all problems the input $P$ will have size
$\wp$.

\item[{\it Output.}]
The multi-way merging problems will be output in sequence,
with the elements of each $W_i$ being output in sorted order;
consequently, the whole output will be in sorted order.
Additionally, for each item $e$ in each $W_i$,
the ranks of the two successive
items straddling $e$ in each of the $r$ input lists forming
the subproblem $W_i$ that contains $e$ are computed.

In other words, each of the
$\wp$ items in the sorted output is listed
along with the ranks for its $r$ pairs
of straddling items, resulting in an output of length
$\wp \cdot (2r+1)$.

\item[{\it Bound.}] This is computed with
$O((x + r) \cdot \wp /B + \wp \cdot r\sqrt{x}/B)$ cache misses.

\item[{\it Method.}]
Let $W$ denote a generic $W_i$, and let $X_1, \cdots X_r$ be the $r$ sorted lists in $W$.\\
$(i)$  For each multi-way merge problem $W$,
for each item $e$ in $W$, perform $r$ binary searches,
one in each of the $r$ lists $X_j$ forming $W$, to find
its rank in the list.\\
Compute the sum of these $r$ ranks to obtain $e$'s rank in $W$.
Note that if an item $e$ has rank $r_j$ in list $X_j$ then its two straddling items in $X_j$
are the items with ranks $r_j$ and $r_j + 1$.\\
$(ii.)$  Reorder the items in $W$ to be in sorted order, keeping the results of its $r$ searches with each item.
To this end, using a Permuting Write with $\ell =1$ and $\ell' = 2r+1$,
the items are copied to their new locations, which are distance $2r+1$ apart;
then, for each item, its $r$ search results are copied.

\item[{\it Analysis.}]
Step $(i)$: If the input for each merging problem fits in cache, i.e.\ if $x \le B^2$, this step incurs
$O(r \cdot \wp /B)$ cache misses (as the cost of writing the output dominates). Otherwise, let
$x_j$ be the length of the $j$-th list in problem $W$, and let $w = |W| = \sum_{j=1}^r  x_j$.
Step $(i)$ incurs $O(w\sum_j \log \ceil{x_j/B}) = O(w^2/B + w\cdot r)$ cache misses.
Summed over all the multi-way
merging problems, this totals $O(x\cdot \wp /B + \wp \cdot r)$ cache misses.
Since $x > B^2$, we have
$\wp \cdot r \le \wp \cdot r \sqrt x/B$.
So the cost is always bounded by
$O(x\cdot \wp /B + \wp \cdot r \sqrt x/B)$.

In Step $(ii)$, the cost of the permuting write for subproblem $W$ is
$O(wr/B + w^2/B)$. Summed over all the subproblems, this totals
$O(\wp \cdot r/B + \wp x/B)$.
The final set of writes incurs
$O(\wp \cdot r/B + \wp )$
 cache misses
(the second term arises due to the reading of each block of $r$ ranks).
If each subproblem fits in cache, i.e.\ if $x\cdot r \le B^2$, then the cost is
$O(\wp \cdot r /B)$,
and otherwise $x \cdot r > B^2$, and then
$\wp \le \wp  \sqrt{xr}/B)$.
So the cost is always bounded by
$O(\wp \cdot r /B + \wp \sqrt{xr}/B)$.

So the overall cache miss cost is
$O((x + r) \wp /B + \wp \cdot r\sqrt{x}/B)$.

\end{itemize}

\subsection{Details of Step 1 in SPMS}
\label{sec:step1-details}

We now give the efficient parallel cache-oblivious algorithm,
assuming the bounds stated for the four procedures defined in Section~\ref{sec:4procs}.
We follow the structure of the algorithm in
Section~\ref{sec:co}, but Step 1a$(ii)$ is changed substantially.

\vone
\noindent
{\bf CO-Parallel-SPMS (Step 1)}\label{sec:par-co-step1}

\vone
\noindent
{\bf Step 1a.} Form and sort the sample set $S$.

\begin{description}
\item
$(i)$ For each of the $r$ sorted lists $L_j$, form the sublist $SL_j$
comprising
every $r^{(c/2)-1}$-th item in $L_j$. The set $S$
comprises the
union of the items in the $SL_j$ lists.
\item
$(ii)$
Output the set $S$ in sorted order in a three-step process as follows.

\begin{itemize}
\item[I]
Form and sort the subarray $SS$ comprising every $r$-th item in the (unsorted) $S$ using a
Small Multi Way Merge.
Note that for each $e$ in $SS$ each pair of straddling items
returned by Small Multi Merge are $r$ apart in the corresponding list forming $S$.

\item[II]  Partition $S$ about $SS$ using a TR Prep followed by a Transposing Redistribution.

\item[III] Sort the resulting subsets of $S$, which all comprise $r$ lists each of at most $r$ items, using a
Small Multi Merge. Note that for each $e$ in $S$, each pair of straddling items are $1$ apart in the corresponding list forming $S$,
and hence $r^{c/2-1}$ apart in the corresponding input list.
\end{itemize}
\end{description}

\noindent
{\bf Step 1b}. Form the pivot set $P$ and use $P$ to partition $A$.

\begin{description}
\item
$(i)$ Form the sorted pivot set $P$ consisting of every $2r$-th item in the sorted $S$.
\item
$(ii)$ Partition the $r$ lists
about $P$ by means of a TR Prep followed by a Transposing Redistribution.
\end{description}

\begin{lemma}
\label{cche-miss-cost-step1}
The sequential execution of the above version of Step 1 incurs $O(n/B)$
cache misses, if $c \ge 6$ and $M \ge B^2$.
\end{lemma}

\begin{proof}
Steps 1a$(i)$ and 1b$(i)$ are simple scans that incur $O(n/B)$ cache misses. The scan used for
forming the subarray $SS$ in Step 1a$(ii)$ is similarly bounded.

In Step 1a$(ii)$ the Small Multi Merge for sorting $SS$ in step I
has only one collection of $r$ sorted lists, and hence
has parameters $x=\wp  = n/r^{c/2}$. Therefore it incurs
$O(n^2/(r^c B) + n/(r^{c/2-1}B) + n^{3/2}/r^{(3c/4 -1} B)) = O(n/B)$ if $c \ge 4$.

In step II, the TR Prep to partition $S$ about $SS$ has parameters $d=r$,
$k = n/r^{c/2}$,
and $y= n/r^{c/2-1}$.
So it incurs $O(n/(r^{c/2-2}B)) = O(n/B)$ cache misses
if $c \ge 4$.
The Transposing Redistribution that follows has $|A| = n/r^{c/2-1}$ and $k= n/r^{c/2}$.
Thus it incurs $O(n/(r^{c/2-1}B) + n/r^{c/2-2}B)) = O(n/B)$ cache misses, if $c \ge 4$.

In step III, the second Small Multi Merge that sorts the subsets of $S$
has parameters $x=r^2$, $\wp = n/r^{c/2-1}$. Therefore it incurs
$O(n/(r^{c/2-3}B) + n/(r^{c/2-2}B) + n/(r^{c/2-3}B)) = O(n/B)$ if $c \ge 6$.

In Step 1b$(ii)$, TR Prep has parameters $d=r^{c/2-1}$,
$k = n/r^{c/2}$,
and $y=n$.
Thus it incurs $O(n/B + n/(r^{c/2-2}B))  = O(n/B)$ cache misses, if $c \ge 4$.
Finally, the second Transposing Redistribution has $|A| = n$ and $k= n/r^{c/2}$.
Thus it incurs $O(n/B + n/r^{c/2-2}/B) = O(n/B)$ cache misses, if $c \ge 4$.
\end{proof}

\begin{theorem}
\label{th:basic-par-spms}
There is an implementation of the SPMS algorithm that performs $O(n\log n)$ operations,
runs in parallel time $O(\log n \log \log n)$,
and incurs $O(\frac nB \frac{\log n}{\log M} + \frac MB \cdot S)$ cache misses in a parallel
execution with $S$ steals, assuming that $M \ge B^2$.
\end{theorem}
\begin{proof}
We will use the implementation of Step 1 given above.
The cache miss bound for Step 1 is $O(n/B)$ by Lemma~\ref{cche-miss-cost-step1}.
Each substep in Step 1 is a fork-join computation which begins with a height
$O(\log n)$ fork tree, and is followed by the complementary join computation, with size $O(1)$
or $O(\log n)$ computations at the intermediate leaves
(the size $O(\log n)$ leaf computations are for the binary searches).
Thus the span (i.e., parallel time) for Step 1 is $O(\log n)$.
An $O(n)$ operation count bound should be immediate.

The bound of $O(\frac nB \frac{\log n}{\log M})$ on the cache misses now follows from
Lemma~\ref{lem:seq-cache-comp}, and the bounds of $O(\log n \log \log n)$ on the parallel time
and $O(n \log n)$ on the operation count from Lemma~\ref{lem:merge-anal}.
Finally, the additive term of $O(MS/B)$ follows from
Corollary~\ref{cor:simple-rws-miss-bound}.
\end{proof}

\noindent
{\bf Comment}. Theorem~\ref{th:basic-par-spms} assumes
a linear space implementation of SPMS.
(Such an implementation would use two collections of size $O(n)$ arrays, and at successive
recursive levels would alternate between which arrays are used for inputs and for outputs,
analogous to what is done in the linear space implementation of the standard merge sort.)
Otherwise, if space was allocated dynamically as needed for each recursive procedure, the algorithm
would use $\Theta(n \log n)$ space, and it would only be for problems of size $O(M/\log M)$ that
their whole computation would fit in cache.
In fact, this would not affect the complexity bounds as the term
$O(\frac nB \frac{\log n}{\log M})$ would be replaced by
$O(\frac nB \frac{\log n}{\log (M/\log M)}) = O(\frac nB \frac{\log n}{\log M})$.

\vhalf
\noindent
{\bf Comparison with the partitioning process in~\cite{BGS09}}.
The sample sort algorithm in
Blelloch et al.~\cite{BGS09}
(which achieves $O(\log^2 n)$ time span deterministically, and $O(\log^{1.5} n)$
span randomized)
also performs a partitioning following the sorting of the sample.
It  used a generalization of matrix transpose to carry out this process;
it plays the same role as our transposing redistribution, however neither can
be used in place of the other due to the different data sizes.
In~\cite{BGS09}, there are $\sqrt n$ sorted subsets each of size $\sqrt n$ to partition about
$\sqrt n$ items, which contrasts with the $r=n^{1/6}$ sorted subsets of total size $n$
being partitioned about $r^3 = \sqrt n$ items.
In turns out we could apply 3 successive iterations of their procedure in lieu of the
transposing redistribution, but this does not appear particularly advantageous.
Further, the method in~\cite{BGS09} can incur significant false sharing misses
in the worst case, in contrast to our method, as discussed in
Sections~\ref{sec:fs-misses}--\ref{sec:fs-estack}.

\section{Bounding False Sharing Costs in SPMS}
\label{sec:fs-misses}

In this section, we address the costs incurred by SPMS due to {\it false sharing},
which we will call {\it fs misses}, and will measure in units of cache miss cost.
Recall that we are working in a multicore setting with $p$ cores,
where each core has a private cache of size $M$, with data organized
in blocks of $B$ words each. When a core needs to access a data item
$x$, this is a unit-cost local computation if the block containing $x$
is in cache, but if the block is not in cache, the core needs to
bring in the data item from main memory at the cost of a
cache miss. This is the traditional caching communication cost that
we have considered in the previous sections, as have all other
papers on multicore algorithms (with the exception of \cite{CR12,CR13}).

Consider the case when cores $C$ and $C'$ both read the same block
$\beta$ into their respective caches in order to access $x\in \beta$ and
$x'\in \beta$ respectively. Suppose $C$ now changes the value in
$x$ with a write. Then, the local view of $\beta$ differs in the caches
of $C$ and $C'$, and this difference needs to be resolved in some
manner. The current practice on most machines is to use a
{\it cache coherence} protocol, which invalidates the outdated copy
of $\beta$ in the cache of $C'$. So, if $C'$ needs to read or write any
word in $\beta$, it will incur a cache miss in order to bring in
the updated copy of $\beta$ into its cache. One can readily design
memory access patterns by two or more cores and associated asynchronous
delays that result in very
large cache miss costs at cores due to invalidations of local blocks
caused by the cache coherence protocol. Under such circumstances,
the block structured storage causes larger delays than that incurred
by reading in individual words at the cost of a cache miss each. This is
considered acceptable in practice because false sharing occurs fairly
infrequently. It is nevertheless considered to be a fairly serious
problem.

Recently, in~\cite{CR12}, systematic algorithm design techniques
were developed that help reduce the cost of false sharing, 
as were techniques specific to the class of
{\it balanced parallel (BP)} and
{\it hierarchical balanced parallel (HBP)} computations.
All of the computations in Step 1 of SPMS are essentially BP computations,
and the overall
SPMS is essentially a `Type 2' HBP computation, so the results in~\cite{CR12}
apply for the most part. However, some of the Step 1 computations (sparse
permuting and parallel binary searches) do not
exactly satisfy the definition of BP,
and since the recursive subproblems in SPMS
do not have approximately the same size (they are only guaranteed to be
sufficiently small), the overall SPMS does not exactly satisfy the HBP
definition. So we need to modify the theorems in~\cite{CR12}.
The modifications needed are small, but as presenting the modifications would
need the necessary background to be provided, we have chosen instead to
provide self-contained proofs of the results for SPMS. The main result
we will establish is that the overall extra communication cost incurred by
the parallel tasks due to false sharing
in an SPMS execution is bounded by the cost of $O(S \cdot B)$ cache misses,
where $B$ is the size of a block, and $S$ is the number of steals
(i.e.,  $S+1$ is the number of  parallel tasks executed).
In other words each parallel task incurs an amortized cost of
at most $O(B)$ cache misses due to false sharing in any parallel
execution of SPMS.

The rest of this section is organized as follows.
In Section \ref{sec:fs-basic} we provide some basic background on
false sharing and results from~\cite{CR12}.
Then in
Section~\ref{sec:fs-spms} we bound the cost of fs misses in SPMS assuming space cannot be
deallocated. Later, in Section~\ref{sec:fs-estack} we re-establish the same bound for the case
when space can be deallocated and then reused on execution stacks.

\subsection{Background on False Sharing}
\label{sec:fs-basic}

An fs-miss occurs if a processor $P$ had a block $\beta$ in its cache and another
processor $P'$ writes a location in $\beta$.
Then in order to access $\beta$ subsequently, $P$ will need to reload $\beta$,
incurring an fs miss.
If there is a sequence of $k$ processors writing to $\beta$ and each one
in turn prevents $P$ from re-accessing $\beta$, then $P$ will be said to incur
$k$ fs misses. It may be that such a sequence can have non-consecutive repetitions,
e.g.\ $P_1, P_2, P_1$ each blocking $P$'s access in turn.
We assume that when an updated $\beta$ is accessed by another processor $P$, the time it takes to
move $\beta$ to $P$'s cache is the time
of $b$ units for a cache miss (or for $O(1)$ cache misses).

We will use the following definition of block delay from~\cite{CR12} to bound the delay due to fs misses.
\begin{definition}
\label{def:block-delay}
\cite{CR12}
Suppose that there are $m\geq 1$ writes to a  block $\beta$ during a time interval $T = [t_1,t_2]$.
Then $m$ is defined to be the\emph{ block delay} incurred by $\beta$ during $T$.
\end{definition}

Consider a processor $P$ that has block $\beta$ in its cache during time interval $T$.
If no other processor writes to $\beta$ during $T$  then $P$ will incur no fs miss on $\beta$. On the other hand,
any time another processor
writes to $\beta$ and $P$ accesses $\beta$ after that write, $P$ will incur an fs miss.
Thus the block delay
$m$ bounds the worst-case number of times that an updated version of the block $\beta$ may need to be moved
into $P$'s cache due to changes made to
the contents of $\beta$ through writes, and thus it bounds the cost of fs misses at $P$ for accessing $\beta$ in
time interval $T$ (in units of cache miss cost).  This is a pessimistic upper bound, since it need not be the case that $P$
will access $\beta$ between every pair of successive writes to it. However, in our analysis, we will use the block delay
$m$ to bound the cost of fs misses at every processor that accesses $\beta$ during the time interval $T$.
As a result, the upper bounds
that we obtain are quite robust, and should apply to most methods likely
to be used to resolve false sharing, cache coherence or otherwise. In the above definition we include the notion of a time
interval because, for example, we do not need to account for writes that occur after a join if all accesses of
$\beta$ by $P$ occur before the join.

We now present the definitions of {\it $L(r)$-block sharing} and
{\it limited access variables} from~\cite{CR12},
and the new notion of buffered arrays (based on the gapping technique in~\cite{CR12});
these provide the algorithmic tools that enable low false sharing costs.

\begin{definition}\label{defn:block-sharing}
\cite{CR12}
A task $\tau$ of size $r$ is \emph{$L(r)$-block sharing},
if there are at most $O(L(r))$ writable blocks that
$\tau$ can share with all other tasks that
could be scheduled in parallel with $\tau$
and that could access a location in the block.
A computation
has block sharing function $L$ if every task in it is $L$-block sharing.
\end{definition}

The next definition specifies the access constraints satisfied by all the variables in the SPMS algorithm.
In part $(i)$ of this definition we reproduce the notion of a limited-access variable from \cite{CR12},
but we introduce the new notion of a {\it buffered} array in part $(ii)$. This leads to a more general notion
of a limited access algorithm defined in part $(iii)$ than the one given in~\cite{CR12} by also including buffered arrays.

\begin{definition}
\label{def:limited-use-var}
\hspace{3in}

\hspace*{-.25in} i.~\cite{CR12}  The variable $x$ is \emph{limited-access} if it is accessed $O(1)$ times
over the course of the algorithm.
\\
ii. An array $A$ is \emph{buffered} if the accesses to the array occur in two phases.
In the first phase, each entry is accessed $O(1)$ times.
The second phase is read-only, and if there are $q$ accesses in the second
phase, 
it is guaranteed that there are no accesses to the initial and to the final
$q$ entries in the array.
\\
iii. An algorithm is \emph{limited access} if its variables are either limited
access or buffered (note that if an array is limited access it need not be buffered).
\end{definition}

\vone
\noindent
\emph{Comment}.
The reason for allowing both $(i)$ and $(ii)$ in a limited access algorithm is that both
lead to the key property of a limited access computation established in 
Corollary~\ref{obs:la-lr}.
We will use buffered arrays to implement arrays on which Parallel-CO-SPMS performs several
binary searches, since the variables stored in these arrays 
may not be limited access variables.
We use the gapping technique~\cite{CR12} to create buffered arrays.
Let $A$ be the array being accessed, and suppose there are at most $q$ accesses to the array.
Then, instead of declaring array $A$, declare array $\overline{A}$ which has a buffer
of $q$ empty locations at its start and end.
It is the SPMS tasks that declare such arrays.
In Parallel-CO-SPMS buffered arrays are needed only in the outputs of the recursive calls,
and in the calls to TR Prep and Small Multi Merge in Step 1 that perform binary searches in parallel.
The outputs of the recursive calls have size linear in their input.
For the binary searches,  the number of accesses is
linear (or less than linear) in the size of the input to this call to SPMS.
Thus, when using buffered arrays for these steps, we maintain the linear space bound
for each call to SPMS.

For now we make a simplifying assumption regarding space allocation,
namely that there is no deallocation.
This ensures that every memory location stores at most one variable.
In Section~\ref{sec:fs-spms}, we analyze the limited access and block sharing
features in SPMS and thereby bound the delay due to fs misses
when there is no deallocation of blocks.
Later, we will extend the analysis to handle the general case in which deallocation can occur.

The following lemma is a  generalization of an observation noted in~\cite{CR12}. Here we
generalize that observation to include buffered arrays, in accordance with our enhanced definition of a
limited access computation in Definition~\ref{def:limited-use-var}. This lemma
shows that an algorithm with limited access, $O(1)$-block sharing, and executed in an environment that has
no deallocation of blocks will have very manageable false sharing costs at each parallel task, even in the worst case.

\begin{lemma}
\label{lem:la-lr}
Consider a parallel
execution of a limited access computation with
no deallocation of blocks,
let $\tau$ be any  parallel task in this execution,
\and let $\beta$ be a block accessed by $\tau$.
The worst case delay incurred by $\tau$
due to fs misses
in accessing $\beta$
is the delay due
to $O(B)$ cache misses, where $B$ is the size of a block.
\end{lemma}

\begin{proof}
If  $\beta$ contains only limited access variables, as it is accessed at most $O(B)$ times,
its block delay is $O(B)$.
If block $\beta$ contains only a (portion of a) buffered array $A$, in Phase 1 for the accesses to
$A$, each location is accessed $O(1)$ times, so there are $O(B)$ accesses, and the block delay
in Phase 1 is $O(B)$. Phase 2 is read-only so there is no further block delay.

Finally we consider the case when $\beta$ contains both a portion of
one buffered array
$A$ overlapping one end (or portions of two buffered arrays $A$ and $A'$  overlapping both ends),
and a collection $C$ of other variables outside of $A$
(and $A'$);
$C$ may include entire buffered arrays.
Note that there are $O(B)$ accesses to the variables in $C$.
Thus we only need to consider Phase 2 accesses to $A$
and $A'$ since there are $O(1)$ accesses to each variable in $A$
and $A'$ in Phase 1.
In this case,  the buffered length $q$ in $A \cap \beta$ is less than $B$,
and $q$ also bounds the number of accesses to $A$ in Phase 2.
An analogous bound applies to the accesses to $A' \cap \beta$.
Hence there are amortized $O(1)$ accesses to each location in $\beta$ when $\beta$ is shared between
portions of one or two buffered arrays,
 where the other variables are either limited access or belong to
entire other buffered arrays.
\end{proof}

\begin{corollary}
\label{obs:la-lr}
Consider a parallel
execution of a limited access computation with $O(1)$-block sharing and
no deallocation of blocks, which incurs $S$ steals.
Then the worst case delay incurred due to fs misses in this parallel execution is the delay due
to $O(S \cdot B)$ cache misses, where $B$ is the size of a block.
\end{corollary}
\begin{proof}
Let $\tau$ be either the original task or a stolen task in this execution, and suppose
that $\tau$ incurs $s$ steals.
Let $\tau_1, \cdots, \tau_j$ be $\tau$'s stolen subtasks.
Since $L(r)=O(1)$,
each $\tau_i$ shares $O(1)$ blocks with
the remainder of the computation.
Thus $\tau$'s stolen subtasks between them share $O(s+1)$ blocks,
counting each block according to the number of times it is shared.
Suppose that there are $k$ such tasks $\tau$, which incur $s_1, s_2, \ldots, s_k$ steals, respectively.
Note that $k \le S+1$.
Now summing over all the tasks shows that there are $O(S + k) = O(S)$ shared blocks.
Hence by Lemma~\ref{lem:la-lr} the worst case
delay incurred by $\tau$ due to fs misses is the delay due to $O(S \cdot B)$ cache misses.
\end{proof}

\noindent
{\bf Discussion.}
Our cost model does not assign a cost for the invalidations
that are performed by the cache coherence protocol of local copies of a block when an entry in
the block is updated. If this delay is larger than a cache miss cost, then we would need to use this larger cost in place of a cache miss cost as the unit of cost for the block delay.
A further refinement arises
 when a large number of parallel reads have been performed on some of the entries in a block $\beta$ in which an entry has been updated.
Then, the cache coherence protocol will need to invalidate $\beta$ in the local caches
of all of the parallel processors that read $\beta$.
If the invalidation cost depends on the number of caches at which a
block needs to be invalidated, then this cost could become large.
In our algorithm, this issue
arises only when performing multiple binary searches on the same array
(in TR Prep and Small Multi Merge).
As already discussed,
we control this potential cost by buffering the arrays undergoing these multiple searches.

\subsection{False Sharing Costs in SPMS: Limited Access and Block Sharing}
\label{sec:fs-spms}

It is readily seen that 
all of the algorithms we described for Step 1
of SPMS are both limited
access and have $L(r) = O(1)$ except for the Binary Search, Transposing Redistribution,
Small Multi Merge, and Sparse Permuting Writes computations.

We now discuss bounds on fs misses for the
routines used in Step 1 that are
not immediately seen as being limited access or having $L(r)=O(1)$.

\vhalf
\noindent
1. {\bf Binary Searches.}

We assume here that each of the binary searches in Step 1 of SPMS
is implemented as a limited access
algorithm, e.g., by using the recursive implementation.
In fact, as we will see subsequently, with a
more precise analysis of where variables are stored,
developed in Section~\ref{sec:fs-estack},
we are able to dispense with this assumption.

 We call binary searches in parallel in our procedures for TR Prep and
Small Multi Merge.
We will use
buffered arrays for the $r$ sorted lists being searched in the binary searches in order to obtain
limited access computations for these steps. This results in a limited access algorithm, and
Corollary~\ref{obs:la-lr}
will apply.

\vhalf
\noindent
2. {\bf Transposing Redistributions.}

Recall that we have
 formulated this
as a log-depth fork-join computation according to a non-standard ordering of the leaves,
namely according to the ordering of the writes. The result of this ordering
is that this computation has
$L(r)=O(1)$, since only the two end blocks of the set of blocks accessed
by a parallel task that is executing part of this computation can be
shared with other parallel tasks.
Had we accessed the data according to the ordering of the reads, we would
not have obtained $O(1)$ block sharing.
Since the variables are limited-access,
Corollary~\ref{obs:la-lr} applies.

\vhalf
\noindent
3. {\bf Small Multi Merge.}

In the Small Multi Merge routine
there is a step that copies over the $2r$ rank values for each element
after the elements have been copied over with a permuting write. As in
Transposing Redistribution, the fork-join tree for this
copying step has its leaves in the output order.
This causes the computation to
have $L(r) = O(1)$.

\vhalf
\noindent
4. {\bf Permuting Writes.}

In our implementation of Permuting Writes
we have a step in which a sequence of $x$ elements are written
into an array $C'$ of size $x$ in an arbitrarily permuted order. (The other steps in this algorithm
are implemented with scans that  are $O(1)$-block sharing and access limited access variables.)
Consider the algorithm that writes
$x$ elements into an array  where the writes could occur in an arbitrarily permuted order.
This algorithm is limited access since each location has only one write.
However, every block in the computation can be shared so the block sharing function
is $L(r)=r$.
This block sharing cost is high,
and only gives a bound of
$O(x \cdot B)$ fs misses, which is excessive.

To reduce the fs miss cost, we modify the Permuting Writes algorithm to
split the writing into two substeps:

\vhalf
\noindent
{\it \underline{Permuting Step in the Permuting Writes Algorithm}.}

\vhalf
\noindent
a.
Initialize an all-zero array
of size $x^2$ and write the $i$-th
output element into location $i \cdot x$. \\
b.
Compact the result in step a into an array of size $x$.

\vhalf
Step b is readily achieved by a log-depth fork-join computation with $L(r) = O(1)$.
We now bound the fs misses incurred by
step a, which performs writes to locations $i \cdot x$,
$1 \le i \le x$, where each of these $x$ locations is written exactly once,
but not in any particular order.

We need to verify that the previous bound on cache misses applies to
this modification of the permuting writes algorithm.
We analyze it as follows.
As before, reading $A$ incurs $O(\ell x/B)$ cache misses.
Now the writes are made to an array $C'$, of size $x^2$.
If $|C'| \le B^2$, then this incurs $O(x^2/B)$ cache misses.
Otherwise, it incurs $O(x)$ cache misses, and as $x > B$ here,
this is also $O(x^2/B)$ cache misses.
The final writing to $C$ incurs $O(x^2/B + \ell' x/B)$ cache misses.
This is a total of
$O(\ell x/B  + \ell' x/B + x^2/B)$ cache misses.

For the fs miss analysis, we only need to consider the case when $x < B$, since there is no
false sharing when
$x \ge B$. Let $ x < B$, and let $i$ be the integer satisfying
$ x\cdot i \leq B <  x \cdot (i+1)$.
Then, at most $i$
writes can occur within a single block.
Each of these writes can incur a cost of at most $i$ cache misses
while waiting to get control of the block.
Hence, the total cost of fs misses
is $O(x \cdot i) = O(x \cdot B/ x) = O(B)$
by our choice of $i$.

The above analysis establishes that for any invocation of a call to
Permuting Writes, the total overhead due to fs misses across all parallel
tasks that participate in this computation is bounded by the cost of
$O(B)$ cache misses. We translate this into an overhead of
$O(B)$ cache misses per parallel task in the computation across all
invocations of  Permuting Writes by assigning the full cost of
fs misses in an invocation to a task stolen during this invocation. Thus,
each parallel task is assigned the fs miss cost of at most one
invocation of Permuting Writes (and if there are no steals in an
invocation, then there is no fs miss during its execution).

Finally, we need to consider the fs misses incurred due to interactions
between blocks accessed for different
calls to Permuting Writes. For this we
observe that although the writes can occur in
arbitrary order in Step 1
within the sparse array in a call to Permuting Writes, only the
two end blocks of the array can be shared by other parallel tasks outside
this call to Permuting Writes, hence the block-sharing function
remains $L(r)=O(1)$.
This leads to the following lemma.

\begin{lemma}
Consider an execution of
SPMS which is implemented with Permuting
Writes.
Suppose that there is no deallocation of space while the algorithm is running.
Then
the total cost for fs misses during the execution of all invocations of
permuting writes is bounded by $O(S' \cdot B)$ cache misses, where $S'$
is the number of steals of subtasks of permuting writes that occur
during the execution.
\end{lemma}

We have shown:

\begin{theorem}\label{thm:fs-start}
Consider a computing environment with tall caches
in which there is no deallocation of space while
an algorithm is running.
There is an implementation of the SPMS algorithm in a parallel execution with $S$ steals
and using a work-stealing scheduler that
performs $O(n\log n)$ operations,
runs in parallel time $O(\log n \log \log n)$,
incurs $O(\frac nB \frac{\log n}{\log M} + \frac MB \cdot S)$ cache misses,
and incurs a total cost due to fs misses delay bounded by $O(S \cdot B)$ cache misses.
\end{theorem}

\noindent
{\bf The Final SPMS Algorithm.}
Our final SPMS algorithm is CO-Parallel-SPMS, described in
Section~\ref{sec:co-parallel}, together with the changes described above for controlling false
sharing costs.

In the next section, we will analyze this final SPMS algorithm to 
show that it achieves the bounds in Theorem~\ref{thm:fs-start} even in a computing
environment where there is deallocation of space on execution stacks.

\section{FS Miss Analysis with Space Deallocation}
\label{sec:fs-estack}

When space for local variables in a computation is reallocated, the same location within a block on an
execution stack could be reused to store different variables, even in a limited access computation.
Thus the bound of $O(B)$ block delay we established in Section~\ref{sec:fs-misses} for any block in a limited access computation
does not necessarily hold.\\
For instance,  as observed in~\cite{CR13},
consider a simple limited access computation $\tau$ consisting of a single
fork-join tree of height $h$, storing one data item at each leaf and at each forking node.
Hence the total space used by this computation on $E_{\tau}$ is $O(h)$ (i.e,
$O(1)$ space for each of up to $h$ segments).
Consider a block $\beta$ on
$E_{\tau}$. Over the duration of this computation, $\beta$ stores data for a
subtree of $\tau$ of depth  $d=O(\min\{B, h'\})$, where $h'$ is the height of the
subtree, and hence $\Theta(2^d)$ different variables
will be stored in $\beta$ during this computation. Thus, although limited access
limits the number of different
values of a variable that can be stored at a given memory
location, due to the re-use of the same location in a block
for different variables, the number of different variables stored on $\beta$
can be even exponential in $B$.

We note that execution under functional programming does not encounter the above issue
since new space is allocated for each new call. For such
environments, our analysis here, outside of usurpations, is not needed. However,
using $E_{\tau}$ in a stack mode is quite common.
In this section we carefully analyze and bound the block delay due to fs misses at a block
on an execution stack of a task in a parallel execution of SPMS.
Section~\ref{sec:fs2-prelim} summarizes some basic facts that will be used in our analysis.
Section~\ref{sec:execfs-blockdelay} bounds the block delay at a single block on an execution stack of a task
in a parallel execution of SPMS.

\subsection{Preliminaries}\label{sec:fs2-prelim}

As in Section~\ref{sec:fs-misses} we will bound the cost of fs misses by bounding the
block delay \cite{CR12} (see Definition~\ref{def:block-delay}) of any block.
We begin by summarizing some basic information from~\cite{CR12}.

\begin{itemize}

\item[(F1)]  {\bf Space Allocation Property}~\cite{CR12}.
 We assume that whenever a processor requests space
it is allocated in block sized units; naturally, the
allocations to different processors are disjoint and entail no block sharing.

{\it Comment.}
We will use the above assumption on how the runtime system allocates blocks in
order to limit the range of time during which moves of a
block on an execution stack across caches can contribute to its block delay for
a particular collection of accesses to it. In particular, if a block used for
an execution stack is deallocated and subsequently reallocated,
accesses to the block in its two different
incarnations do not contribute to each other's block delay.

\item[(F2)] {\bf Block Move Fact}~\cite{CR12}.
Let $\beta$ be a block on the execution stack $E_{\tau}$ of a task $\tau$,
and let $T'$ be any subinterval of time during which $\tau$ is
executed.
Suppose that the processors
executing stolen 
 subtasks of $\tau$ during $T'$ access block $\beta$ a total
of $x$ times during $T'$.
If there are $u$ usurpations of $\tau$ during $T'$,
then $\beta$ incurs a block delay of at most $2x + u$ during $T'$, regardless of
the number of accesses to $\beta$ made during $T'$ by the
processor(s) executing $\tau$.

{\it Comment.}
The above fact is given in~\cite{CR12}.
It simply observes that the processor $C$ which currently ``owns'' $\tau$'s execution stack, that
is the processor which can add and remove variables from the stack,
can make
many accesses to $\beta$ but contributes
toward its block delay only if $\beta$ has
moved to a processor executing a stolen 
 task, and hence needs to move back to
$C$'s cache at the cost of adding 1 to its block delay.

\item[(F3)] {\bf The Steal Path Fact (from Section~\ref{sec:model}).}
Let $\tau$ be either the original task or a stolen 
 subtask.
Suppose that $\tau$
incurs steals of subtasks $\tau_1, \cdots, \tau_k$.
Then there exists a path $P_{\tau}$ in $\tau$'s computation dag from its
root to its final node such that the parent of every stolen 
 $\tau_i$ lies on $P_{\tau}$, and every off-path right child of
a fork node on
$P$ is the start node for a stolen  subtask.

{\it Comment.}
This fact is reproduced from Section~\ref{sec:model}.
\end{itemize}

The following lemma is a simple consequence of the above assertions.

\begin{lemma}
\label{lem:beta}
Let $\tau$ be a limited access computation that incurs some steals,
and let $\beta$ be a block allocated on
$\tau$'s execution stack $E_{\tau}$.
Let $P_{\tau}$ be the steal path in $\tau$. Then,

\begin{enumerate}
\item If $\beta$ is used only for segments of nodes
that are not on $P_{\tau}$, then it incurs no fs misses.
\item
If $\beta$ stores some data for one or more segments for nodes on $P_{\tau}$,
then
its block delay is
$O(V + U)$, where $V$ is the number of different variables used in the
segments for nodes on $P_{\tau}$ stored in
$\beta$, and $U$ is the number of usurpations
of $\tau$.
\end{enumerate}
\end{lemma}

\begin{proof}
For the first part  we observe that since $\beta$ does not store data for
any segment on $P_{\tau}$,  it is not  a shared block on the execution stack,
and hence it incurs no fs misses.

For the second part,
we note that, between them, the variables in the segments for nodes
on $P_{\tau}$ are accessed $O(V)$ times.
But these are the only variables that stolen 
 subtasks can access.
Thus, by Fact (F2), the overall block delay is $O(V+U)$.
\end{proof}

\subsection{FS Misses at a Single Block on an Execution Stack}\label{sec:execfs-blockdelay}

\subsubsection{FS Misses on the Execution Stack during Step 1}

Facts F1--F3 and Lemma~\ref{lem:beta}
provide us the basic tools for
bounding the block delay of a block on an execution stack.
We begin by presenting this bound for
the fork-join tree computations in Step 1 of SPMS; in the next section
we present the bound for the overall SPMS algorithm. The following lemma is very similar to
Lemma V.4 in~\cite{CR12}, but we have generalized it
to state the result in terms of the
total space used, allowing for non-constant space usage
at all nodes, and relaxing the requirement of limited access for the computation at the leaves of the fork-join tree.

\begin{lemma}
\label{lem:fs-fork-join}
Let $\tau$ be a fork-join computation, where both the internal nodes and the leaf nodes may
use non-constant space, and let the total space used by $\tau$ be $\Sigma$.
If the accesses to the segments for the internal nodes of $\tau$ are limited access then in
any parallel execution of $\tau$,
any block $\beta$ allocated on $\tau$'s execution stack incurs a block delay of
$O(\min\{B, \Sigma\})$ cache misses during $\tau$'s execution.
\end{lemma}

\begin{proof}
We will apply Lemma~\ref{lem:beta}, so we need to bound $V$ and $U$, considering only the variables
used in segments stored on the steal path $P_{\tau}$. We first consider $V$.
If we consider a single path from the root to a descendant node in a  fork-join tree, there is no space reuse on the execution stack,
and  the segments for the fork nodes along this path will occupy distinct but contiguous locations on the execution stack.
Since the total space used by $\tau$ is $\Sigma$, this also bounds the space used by all segments for the fork nodes on $P_{\tau}$.

It is possible that a segment for a leaf node $l$ lies on $\beta$, and since the accesses at a leaf node need not
be limited access
(if it corresponds to a non-recursive implementation of a binary search),
there is no bound on the number of accesses to $\beta$ made by the computation at $l$.
But since the computation at a leaf of a fork tree
is always executed by the processor that executes
the task kernel to which the leaf belongs we have by Fact (F2) that any fs miss cost for the accesses to variables on the execution stack for a leaf node
is absorbed into the cost for fs misses at $\beta$ for stolen  tasks.
Hence $V = O(\min\{B, \Sigma\})$.

We bound $U$ as in~\cite{CR12}. Usurpations occur during the up-pass, and only during accesses to
segments on the steal path $P_{\tau}$. There is at most one usurpation per segment stored on $\beta$,
hence we can bound
$U=  O(\min \{B, h\})$, where $h$ is the height of the fork-join tree for $\tau$. Since each node stores at least
one element for its segment, $h \leq \Sigma$, hence $U=  O(\min \{B,\Sigma\})$.
The lemma now follows from Lemma~\ref{lem:beta}.
\end{proof}

\begin{lemma}
\label{lem:step-1-space}
Each Step 1 task $\tau$ uses space
$\Sigma = O(\log |\tau|)$.
\end{lemma}
\begin{proof}
Let $\tau$ be a Step 1 task.
The sequence of segments on its execution stack $E_{\tau}$ comprises
$O(\log|\tau|)$ segments for the fork nodes on the path to any leaf node in $\tau$,
followed by the segments for the leaf node (in the case of a
recursively implemented
binary search, the leaf node corresponds
to a recursive computation of length $O(\log |\tau|)$ and thus induces a sequence of $O(\log |\tau|)$ segments;
otherwise the leaf computation has length $O(1)$ and induces a single segment).
Each segment has size $O(1)$ and thus the total size of the segments is $O(\log |\tau|)$.
\end{proof}

Each substep in Step 1 of SPMS comprises either
a sequence of $O(1)$ size $n$ fork-join tasks, each
of depth $O(\log n)$, or a collection of
such tasks which are performed in parallel,
and which have combined size $O(n)$.
The initiation of these tasks is itself performed by a fork join-tree of depth $O(\log n)$.
Each fork-join tree uses constant space for each segment, and the leaves use constant space,
except for the binary searches which,
use $O(\log n)$ space; hence $\Sigma = O(\log n)$.
As there are $O(1)$ substeps in Step 1, it follows that Step 1 has depth $O(\log n)$.

Thus, we have the following corollary.

\begin{corollary}
Let $\tau$ be a task executed in Step 1 of SPMS,
and let $\beta$ be a block on its execution stack.
Then the block delay incurred by $\beta$ is bounded by
$O(\min\{B, \log |\tau|\})$.
\end{corollary}

\subsubsection{FS Misses on the Execution Stack for the Whole SPMS Algorithm}

The following lemma bounds the
number of different variables used by a task $\tau$ in
the SPMS algorithm of Section~\ref{sec:fs-misses}
(i.e., the cache-oblivious, parallel version of SPMS, with the
small enhancements made in Section~\ref{sec:fs-misses}  to reduce fs miss costs).
We bound the number of variables rather than the space used by $\tau$
because the same space may be allocated in turn for the segments
of successive recursive calls.

\begin{lemma}
\label{fs-space}
Let $\tau$ be a task in an execution of SPMS. Then, the
sum of the sizes of all segments placed on $E_{\tau}$,
over the course of the complete execution of $\tau$, for nodes on any single
path in $\tau$'s computation dag, is $O(|\tau|)$.
\end{lemma}

\begin{proof}
If $\tau$ is a Step 1 fork-join computation,
or a subtask of a Step 1 fork-join computation
then by Lemma~\ref{lem:step-1-space} the total size of the segments on
any path in
$E_{\tau}$ is $O(\log |\tau|) = O(|\tau|)$,
because the sequence of segments are simply allocated and then deallocated; there is
no reuse of space in this case.

If $\tau$ is a recursive call in Step 2 or 3,
the first segment for the nodes on a path in $\tau$'s computation dag
is the segment for the variables $\tau$ declares, which has size $O(|\tau|)$;
this is followed by
the segments for paths in $\tau$'s
 Step 1 subtasks; there are $O(1)$
such subtasks and they each induce a sequence of segments as described in the previous paragraph,
which therefore have total size $O(\log|\tau|)$.
This is followed by the segments for $\tau$'s Step 2 subtask, which in turn is followed by the segments for
its Step 3 subtask. These are identical in structure, so we will describe the segments for Step 2 only.
First, there is path of length $O(\log r)$ in the fork tree for the $O(r^{c/2 + 1/2})$ Step 2 recursive subtasks,
and each node on this path has
a segment of constant size.
This is followed by the segments for a recursive task
similar to $\tau$ but having size $O(\sqrt{|\tau|})$.
Thus the total size of the segments on $E_{\tau}$ is given by
$S(|\tau|) = |\tau| + 2 S(\sqrt{ |\tau|})$
and $S(1)  = O(1)$ which has solution
$S(|\tau|) = O(|\tau|)$.
\end{proof}

Consider the execution stack $E_{\tau}$ for the original task or a recursive sorting task $\tau$.
At any point in time $E_{\tau}$
will contain a sequence of segments for the nodes that have started execution
but not yet completed.\\
At the bottom is the segment $\sigma_{\tau}$ for the variables declared by $\tau$.
This is followed by a (possibly empty) sequence of segments for a series of recursive SPMS tasks,
where for each recursive task $\mu$ there are a series of segments for the fork nodes on the path to $\mu$
followed by a segment for the variables declared by $\mu$;  finally,
at the top there may be a sequence
of segments for a Step 1 subtask, if the computation is currently executing Step 1 in some recursive SPMS subtask.

We now bound the block delay of a block $\beta$ on $E_{\tau}$. If $\tau$ incurs no steal then
$\beta$ incurs no fs misses.
If $\tau$ incurs at least one steal then, by
the Steal Path Fact (F3), there is a unique steal path $P_{\tau}$ such that every
stolen 
 task is the right child of a node on $P_{\tau}$. Further, the
accesses that each such stolen 
 task makes to $E_{\tau}$ are to the segment
corresponding to its parent node, which lies on $P_{\tau}$. Thus fs misses
incurred on the execution stack will all correspond to accesses to blocks that
contain segments for parents of stolen 
 nodes; all of these
parent nodes will lie on $P_{\tau}$. Thus, by the Block Delay Fact (F2),
it suffices to bound the accesses to
$\beta$ for segments placed in it that lie on $P_{\tau}$.
We use these observations in our proof of the following lemma.

\begin{lemma}
\label{lem:blck-del}
Let $\beta$ be a block on the execution stack $E_{\tau}$ of a task $\tau$
in a parallel execution of SPMS.
Then the block delay of $\beta$ is bounded by $O(\min\{B, |\tau|\})$.
\end{lemma}

\begin{proof}
By Lemma~\ref{lem:beta},
it suffices to consider accesses to $\beta$ while it stores some
element in a segment on $P_{\tau}$.
 If $\beta$ is currently of this form, then the
bottommost segment it stores, call it $\sigma$, must lie on $P_{\tau}$,
since $P_{\tau}$ starts
at the root of the computation dag for $\tau$.
Let $\rho$ be the subcomputation of $\tau$ which declared the variables stored
in segment $\sigma$.
This instantiation of block $\beta$ will remain allocated until $\rho$ completes
its computation.
We will now bound the number of accesses to $\beta$
on $P_{\tau}$ during $\rho$'s execution.
By the Block Move Fact (F2), we only need to consider accesses to segments on $P_{\tau}$
by stolen 
 tasks, and then add in $U$,
where $U$ is the number of usurpations of $\tau$ that occur in this interval.

If $\rho$ is a Step 1
fork-join computation, then by Lemma \ref{lem:fs-fork-join} the block delay
for $\beta$ during the duration of $\rho$ is
$O(\min\{B, \log |\rho|\})$,
and we are done
(since $\beta$ will be deallocated once $\rho$ completes its computation).
Now suppose $\rho$ is a recursive call in Step 2 or Step 3,
or a collection of such recursive calls initiated at a fork node.
Let $P_{\rho} = \rho \cap P_{\tau}$.

If $|\rho|=O(B)$ then by Lemma~\ref{fs-space} the sum of the sizes of the segments
for the computations for nodes on $P_{\rho} = \rho \cap P_{\tau}$
 is also $O(B)$,
hence $O(B)$ accesses to $\beta$ by stolen 
nodes occur before $\beta$ is deallocated, and the block delay for $\beta$ is $O(B)$.

If $|\rho|=\Omega(B)$,
then the following amount of data for nodes on $P_{\rho}$ is placed on $\beta$:
$O(\min\{\log |\rho|, B\})$ data for the nodes forking the recursive calls (in the case that $\rho$ consists
of multiple recursive calls initiated at a fork node in some Step 2 or Step 3 computation);
$O(B)$ data for each of the constant number of
fork-join computations in Step 1 of $\rho$, followed by
$O(\min\{\log |\rho|, B\})$
for the forking of recursive subtasks in Step 2, followed
by the segments for a recursive SPMS task in this same Step 2, followed
by a similar sequence of segments for Step 3.
We analyze the Step 2 cost closely; the same analysis applies to the cost of Step 3.
If the segments for the fork-tree for the recursive
tasks fill up $\beta$, then the segments due to the recursive task lie outside of $\beta$,
and so the total contribution of Step 2 is an additional $O(B)$ data.
If the fork-join tree does not fill up $\beta$, we will also have to handle
the size $O(\sqrt{ |\rho|)}$ recursive SPMS call on $P_{\tau}$.
Here again,
there are two cases: One is that the segment for the recursive subproblem fills up
$\beta$; in this case  the overall amount of data is $O(B)$.
Finally, if the segment for  the recursive call does not fill up $\beta$, then it
has size less than $B$, and hence, by Lemma~\ref{fs-space} the total
size of segments in $\beta \cap P_{\rho}$
during Step 2 is $O(B)$.
Thus, $O(B)$ variables are stored in $\beta$ and hence
the $V$ component of $\beta$'s block delay is $O(B)$.

To complete the proof we need to bound $U$, the number of usurpations that are
followed by an access to $\beta$, and that occur during $\rho$'s duration.
When a usurpation of $\tau$ occurs, the processing of $\tau^K$ had already reached
the join node $v^J$ for the usurping subtask. Consequently, at this time,
the segments on $E_{\tau}$ are all for nodes on $P_{\tau}$.
Thus, when the usurping processor $C'$ first accesses $\beta$ it must either access a segment
currently on $P_{\tau}$, or it must remove the topmost segment on $P_{\tau}$ that
overlaps $\beta$ (in order that it can place a new segment on $\beta$ which will be
the location for its first access to $\beta$).
But as already shown, only $O(B)$ data for nodes on $P_{\tau}$ is stored on $\beta$;
so the first possibility can occur only $O(B)$ times.
And as there is only $O(B)$ data among all the segments of $P_{\tau}$ stored on $\beta$,
there can be at most $O(B)$ such segments, and hence the second possibility can also
occur only $O(B)$ times.
Thus, $U = O(B)$.
Hence, the block delay of any block on $E_{\tau}$ is $O(\min\{B, \tau\})$.
\end{proof}

We can now establish Theorem~\ref{thm:fs-full}.

\begin{proof}
(of Theorem~\ref{thm:fs-full}).
Because of the $L(r) = O(1)$ bound, each steal and usurpation results in a further $O(1)$ blocks
being shared (and if a block is shared by $k$ such stolen and usurped tasks, it will be counted
$k$ times).
Note that there is at most one usurpation per steal.
Thus there are $O(S)$ sharings, each incurring an $O(B)$ delay due to fs misses
by Lemma~\ref{lem:blck-del}, for a total delay of $O(S \cdot B)$.
The other bounds remain unchanged from Theorem~\ref{thm:fs-start}.
\end{proof}

\section{Overall Performance}\label{sec:perf}

In this section we bound the overall running time of the SPMS algorithm.
First, using Equation~\eqref{eqn:rws-time}, the running time of SPMS executing on $p$ processors,
when there are $S$ steals is:

\begin{eqnarray}\label{eqn:spms-time}
T_p & = & O\left( \frac{1}{p} \left[ n \log n + b \cdot \frac{n}{B} \cdot \frac{\log n}{\log M}  + b\cdot S\cdot \frac{M}{B} + T_s  +  T_u  + I \right] \right)
\end{eqnarray}

Note that the term $F(S)= S \cdot B$ for the false sharing cost is upper bounded by
the cache miss overhead
since we have a tall cache, and hence does not appear as a separate term.

We can apply the above equation to any work-stealing scheduler for which we have bounds
on the number of steals. This is the case for RWS, which we consider next.

\subsection{Performance of SPMS Under RWS}

Bounds on the number of steals under RWS have been obtained in a sequence of papers,
first in Blumofe and Leiserson~\cite{BL99},
considering only the number of parallel steps,
then  in Acar et al.~\cite{ABB02} considering the overhead of standard caching costs,
and more recently in Cole and Ramachandran~\cite{CR13}
 considering the overhead of both caching and fs miss costs. We will now apply the results in~\cite{CR13}
to obtain a w.h.p.  bound on the running time of SPMS under RWS.

\begin{theorem}
\label{thm:rws-perf}
Suppose SPMS is executed on $p$ processors using the RWS scheduler.
Then, w.h.p.\ in $n$, if $M= \Omega (B^2)$, it runs in time
\begin{equation*}
O\left(
\frac 1p \left[ n\log n +   b \frac {n}{B} \frac{\log n}{\log M} \right]
+  \left[b \cdot \frac MB + s\right] \log n \cdot \log\log n +\frac{b^2}{s} \cdot \frac {M}{\log B} \cdot \log n
\right).
\end{equation*}
\end{theorem}

\begin{proof}
We will substitute appropriate values in Equation~\eqref{eqn:spms-time}.
The analysis in~\cite{CR13} of the RWS scheduler on HBP algorithms, considering both caching and
fs miss costs,
applies to the SPMS algorithm and shows that, w.h.p.\ in $n$,
$S = O(p \cdot [ T_{\infty} + \frac bs B \frac {\log  n} {\log B}])
= O(p \cdot [\log n \cdot \log\log n + \frac bs B \frac {\log  n} {\log B}])$.
It also shows that,  w.h.p.\ in $n$,
the total time $T_s + T_u$ spent on steals,
both successful and unsuccessful, is
$O(s\cdot p \cdot T_{\infty} ) = O(s\cdot p \cdot \log n \cdot \log\log n )$.
Finally, we note that $I=0$ as any processor without a task is always trying to steal
under RWS, and hence
there is no idle time not already accounted for.
\end{proof}

We can now establish Theorem~\ref{thm:rws}.

\begin{proof}
(of Theorem~\ref{thm:rws}).
The first term in Theorem~\ref{thm:rws-perf} represent the optimal performance in terms of operation count
and cache misses
when $M = \Omega(B^2)$.
Thus, for optimality, it suffices to have the
second and third terms bounded by the first term.
As, by assumption, $s = O(b \cdot \frac MB)$, the second term reduces
to $O( b \cdot \frac MB  \cdot \log n \cdot  \log\log n)$.
Rearranging yields the
desired bound.
\end{proof}

\section{Conclusion}
\label{sec:discussion}

We have presented and analyzed SPMS, a new deterministic  sorting algorithm. The final version of SPMS in Section~\ref{sec:fs-misses}
simultaneously achieves optimal $O(n \log n)$ sequential running time, optimal $O(\frac{n}{B} \cdot \log _M n)$ sequential cache misses
(assuming a tall cache as required for this bound),
and a critical path length of $O(\log n \cdot \log\log n)$, which is almost optimal. In a parallel execution under a work stealing
scheduler, this algorithm has been designed to have
a low false sharing cost bounded by $O(S \cdot B)$, where $S$ is the number of steals. Its 
traditional cache miss overhead in such a parallel execution is $O(S \cdot (M/B))$ using well-known earlier results. These bounds
are resource oblivious: the algorithm itself uses no machine parameters, and the dependence on $M$ and $B$ are determined
only by the analysis of its performance. SPMS uses
a new approach to sorting, by combining elements from both Mergesort and Samplesort.

 An important open question that remains is whether resource oblivious sorting can be performed in $O(\log n)$ critical path length,
 while maintaining optimal sequential cache complexity. The challenge here is the requirement that none of the parameters
  $p$,  $M$ and $B$
can be used within the algorithm.
  This is a topic for further investigation.

\bibliographystyle{abbrv}
\bibliography{sort,rws-refs}

\end{document}